\documentclass{article}  
\usepackage[margin=1 in]{geometry}
\usepackage{authblk,amsthm,amssymb}
\usepackage{blindtext}
\usepackage{hyperref}
\usepackage{graphicx}
\usepackage{amsfonts} 
\usepackage{pifont}
\usepackage{subcaption}
\usepackage{mathtools}
\def\A{{\cal A}}
\def\I{{\cal I}}
\def\P{{\cal P}}
\def\s{{\mathcal S}}
\def\OO{{\cal O}} 

\DeclareMathOperator{\interior}{int}

\newcommand{\GS}{\textsf{\sc Game-Of-Same-Side}}

\newcommand{\IR}{\mathbb{R}}
\newcommand{\AC}{\textsf{\sc Algorithm-Center}}
\newcommand{\AV}{\textsf{\sc Algorithm-Vertex}}

\DeclarePairedDelimiter\abs{\lvert}{\rvert}%
\DeclarePairedDelimiter\norm{\lVert}{\rVert}%
\let\oldabs\abs
\def\abs{\@ifstar{\oldabs}{\oldabs*}}
\let\oldnorm\norm
\def\norm{\@ifstar{\oldnorm}{\oldnorm*}}

\makeatletter
\DeclareFontFamily{U}{tipa}{}
\DeclareFontShape{U}{tipa}{m}{n}{<->tipa10}{}
\newcommand{\arc@char}{{\usefont{U}{tipa}{m}{n}\symbol{62}}}%

\newcommand{\arc}[1]{\mathpalette\arc@arc{#1}}

\newcommand{\arc@arc}[2]{%
  \sbox0{$\m@th#1#2$}%
  \vbox{
    \hbox{\resizebox{\wd0}{\height}{\arc@char}}
    \nointerlineskip
    \box0
  }%
}
\makeatother

\newtheorem{property}{Property}
\newtheorem{corollary}{Corollary}
\newtheorem{clm}{Claim}

\newtheorem{definition}{Definition}
\newtheorem{theorem}{Theorem}
\newtheorem{lemma}{Lemma}

\newtheorem{observation}{Observation}

\usepackage{color}

\newcommand{\myqed}{\hfill $\Box$}

\begin{document}
\title{Online Geometric Covering and  Piercing\thanks{Preliminary version of this paper has appeared in the proceedings of the 42nd IARCS Annual Conference on
Foundations of Software Technology and Theoretical Computer Science (FSTTCS), 2022~\cite{DeJKS22}.}}

\author[1]{Minati De\footnote{Partially supported by SERB-MATRICS grant MTR/2021/000584.}}
\affil{Deptartment of Mathematics\\ Indian Institute of  Technology Delhi, India\\
\texttt{\{minati,satyam.singh$^{\star}$\}@maths.iitd.ac.in, sakshampv@gmail.com, saratvarmakallepalli@gmail.com }}
 
\author[1]{Saksham Jain}
\author[1]{Sarat Varma Kallepalli}
\author[1]{Satyam Singh$^{\star}$\footnote{Supported by CSIR (File Number-09/086(1429)/2019-EMR-I).

$\ ^{\star}$Corresponding author}}

\maketitle

\begin{abstract}
We consider the online version of the piercing set problem, where geometric objects arrive one by one, and the online algorithm must maintain a valid piercing set for the already arrived objects by making irrevocable decisions.
It is easy to observe that any deterministic algorithm solving this problem for intervals in $\IR$ has a competitive ratio of at least $\Omega(n)$. 
This paper considers the piercing set problem for similarly sized objects.
We propose a deterministic online algorithm for similarly sized fat objects in $\IR^d$.
 For  homothetic hypercubes in $\IR^d$ with side length in the range  $[1,k]$, we propose a deterministic algorithm having a competitive ratio of at most~$3^d\lceil\log_2 k\rceil+2^d$.  In the end, we show deterministic lower bounds of the competitive ratio for similarly sized $\alpha$-fat objects in $\IR^2$ and homothetic hypercubes in $\IR^d$. Note that piercing translated copies of a convex object is equivalent to the unit covering problem, which is well-studied in the online setup.
Surprisingly, no upper bound of the competitive ratio was known for the unit covering problem when the corresponding object is anything other than a ball or a hypercube. Our result yields an upper bound of the competitive ratio for the unit covering problem when the corresponding object is any convex object in $\IR^d$.\\

\end{abstract}

\section{Introduction}
\pagenumbering{arabic}
Piercing is one of the most important problems in computational geometry. 
A set ${\P}\subseteq \IR^d$ of points is defined as a \emph{piercing set} for a set ${\cal S}$ of geometric objects in $\IR^d$ if each object in ${\cal S}$ contains at least one point from ${\P}$. Given a set $\cal S$ of objects, the problem is to find a piercing set ${\P}$ of the minimum cardinality.

In this paper, we study the \emph{online version} of the problem in which  geometric objects arrive one by one, and the online algorithm needs to maintain a valid piercing set for the already arrived objects.
Once a new geometric object $\sigma$ arrives, the online algorithm needs to place a point $p\in \sigma$  if the existing piercing set does not already pierce it. Note that an online algorithm may add points to the piercing set
but cannot remove points from it, i.e., the online algorithm needs to make irrevocable decisions. The goal of the problem is to minimize the cardinality of the piercing set.
We analyze the quality of our online algorithm by competitive analysis~\cite{BorodinE}. The \emph{competitive ratio} of an online algorithm is  $\sup_{\beta}\frac{\A_{\beta}}{\OO_{\beta}}$, where $\beta$ is an input sequence, and $\OO_{\beta}$ and $\A_{\beta}$ are the cost of the solution produced by an optimal offline algorithm and the online algorithm, respectively, for
the input sequence $\beta$. Here, the supremum is taken over all possible input sequences $\beta$.

An application in wireless networks inspires the online version of this problem~\cite{EvenS14}. Here, the points
model base stations and the centers of objects model clients. The reception range of each client has some geometric shape (e.g., disks, hexagons, etc.). The algorithm must place a base station serving a new uncovered client. Since installing the base station is expensive, the overall goal is to place the minimum number of base stations.

\subsection{Related Work}
 In the offline setting, the piercing set problem is a well-studied problem~\cite{Chan03,EfratKNS00,KatzNS03}. For one-dimensional intervals, a minimum piercing set can be found in $O(n\log{ c})$ time, where $n$ is the number of one-dimensional intervals and $c$ is the size of a minimum piercing set~\cite{Nielsen00}.
On the other hand, computing a minimum piercing set is  NP-complete for unit squares~\cite{GareyJ79}.  Chan~\cite{Chan03} proposed a polynomial-time approximation scheme for $\alpha$-fat convex objects. Though the piercing set problem is well-studied in the offline setup, surprisingly, there is a lack of study for this problem in other models. Katz et al.~\cite{KatzNS03} studied this problem for one-dimensional intervals in the dynamic setup.  In this setup, the object can arrive as well as depart, but unlike the online model, here, the decisions can be reverted. For a set $\mathcal{S}$  of intervals in $\IR$, they proposed a linear-size dynamic data structure that allows computing a new minimum piercing set in $O\left(c(\mathcal{S}) \log |\mathcal{S}| \right)$ time per insertion or deletion, where $c(\mathcal{S})$ is the size of a minimum piercing set for $\mathcal{S}$.

The set cover and hitting set problems are closely related  to the piercing set problem. Let $\P$ be a set of elements, and let $\cal S$ be a family of subsets of $\P$. A \emph{set cover} is a subset $\cal C \subseteq \cal S$ such that the union of sets in $\cal C$ covers all elements of $\cal P$ and a \emph{hitting set} is a collection of points $\cal H\subseteq\P$ such that the set $\cal H$ intersects every set $s$ in $\cal S$. The set cover (respectively, hitting set) problem aims to find a set cover $\cal C$ (respectively, hitting set $\cal H$)  of the minimum cardinality. 
If $\P=\IR^d$ and $\cal S$ is a family of objects in $\IR^d$, then the corresponding hitting set problem is the same as the piercing set problem.
It is well known that a set cover of the tuple $({\cal P}, {\cal S})$ is a hitting set of the tuple $({\cal P}^{\perp}, {\cal S}^{\perp})$~\cite{AgarwalP20}. Here, for each set $s\in{\cal S}$ there is an element in ${\cal P}^{\perp}$ and for each element $p\in{\cal P}$ there is a set $s_p$, namely, $s_p=\{s\in{\cal S}\ |\ p\in s\}$, in ${\cal S}^{\perp}$. 
 In the offline setup, if $\cal P\subset\IR$ and the set $\cal S$ consists of intervals in $\mathbb{R}$, the set cover problem can be solved in polynomial time. However, both the set cover and hitting set problems are  NP-hard, when $\cal S$ consists of simple geometric objects like unit disks in $\mathbb{R}^2$ and ${\cal P}\subset{\IR}^2$~\cite{FowlerPT81}. Alon et al.~\cite{Alon09} initiated the study of the online set cover problem. In their model, the finite sets $\P$ and $\cal S$ are already known. However, the order of arrival of points from the set $\P$ is unknown. Upon the arrival of an uncovered point in $\cal P$, the online algorithm must
choose a set $s\in \cal S$ that covers the point.
The algorithm presented by Alon et al.~\cite{Alon09} has a competitive ratio of $O(\log n \log m)$, where $|\P|=n$ and $|{\cal S}|=m$.
Recently, in the same model, Khan et al.~\cite{KhanLRSW23} proposed an algorithm having an optimal competitive ratio of $\Theta(\log n)$ for the set cover problem when $\cal S$ consists of homothetic (translated and scaled) squares and
$\P$ consists of $n$ points in $\IR^2$.
Even and Smorodinsky~\cite{EvenS14} studied the online hitting set problem, where both sets $\P$ and $\cal S$ are known in advance, but the order of the arrival of input objects in $\cal S$ is unknown. In this model, they proposed  an algorithm having an optimal competitive ratio of $\Theta(\log n)$ for intervals, half-planes or unit disks, where $n$ is the cardinality of the set $\cal S$. 
Recently, Khan et al.~\cite{KhanLRSW23} studied the hitting set problem, where for a given positive integer $N$, the point set $\P$ is a subset of integral points from $[0,N)^2$ and $\cal S$ consists of axis-parallel squares in $\IR^2$ whose vertices are from $[0,N)^2$. They obtained an optimal $\Theta(\log N)$-competitive algorithm for this variant.
{Extending the result of Even and Smorodinsky~\cite{EvenS14}}, recently, De et al.~\cite{DeMS23}  studied  the set cover (respectively, hitting set) problem in an online model where only the set $\cal S$ (respectively, $\P$) is known in advance. The set $\P$ (respectively, $\cal S$) is unknown apriori, and the elements of the set $\P$ (respectively, $\cal S$) are revealed one after another. In this model, they propose an optimal $\Theta(\log n)$-competitive algorithm when objects in $\cal S$ are   translated copies of a  regular $k$-gon ($k\geq 4$) or a disk, and $\P$ consists of points in $\IR^2$.
When $\P=\mathbb{Z}^d$ and $\cal S$ is a family of geometric objects in $\IR^d$, the corresponding online hitting set problem is studied in~\cite{DeS22}.

A related problem is the \emph{unit covering} problem (a special variant of the set cover problem), where $\P$ is a set of points, and the set $\cal S$ consists of all (infinite) possible translated copies of a given object $s$.  Note that when the object $s$ is convex, the unit covering problem is equivalent to the special case of the piercing set problem, where objects are translated copies of a convex object $-s$ (see Section~\ref{Unit_cover}). In the online version of the unit covering problem, the set $\cal S$ is not known in advance. Charikar et al.~\cite{CharikarCFM04} studied the online version of the unit covering problem, where $s\subset \IR^d$ is a unit ball.  
 They proposed an online algorithm having a competitive ratio of~$O(2^dd\log d)$. They also proved $\Omega(\log d/\log\log \log d)$ as the deterministic lower bound of the competitive ratio for this problem.
 Dumitrescu et al.~\cite{DumitrescuGT20} improved both the upper and lower bounds of the competitive ratio to $O({1.321}^d)$ and $\Omega(d+1)$, respectively.  When $s\subset \IR^d$ is a centrally symmetric convex object, they proved that the competitive ratio of every deterministic online algorithm is at least $I(s)$, where $I(s)$ is the illumination number (for definition, see Section~\ref{Notations}) of the object $s$. 
 Recently, Dumitrescu and T{\'{o}}th~\cite{DumitrescuT22} proved that the competitive ratio of any deterministic online algorithm for the unit covering problem is at least $2^d$  when $\P=\IR^d$ and {$ s$ is a hypercube in $\IR^d$}. 
Surprisingly, we did not find any upper bound of the competitive ratio {when  $s$ is anything other than a ball or a hypercube}. In this paper, we raise this question and obtain a very general upper bound for the same.

Another relevant problem is the \emph{chasing convex objects}, initiated by Friedman and Linial~\cite{FriedmanL93}. 
Here, similar to the piercing set problem,  convex object $\sigma_i$ arrives {at time $i$}, and the online algorithm needs to place a piercing point $p_i\in \sigma_i$, and it pays a cost of the distance between the currently placed point $p_i$ and the last placed point $p_{i-1}$, i.e., $d(p_i,p_{i-1})$.
Note that the objective of this problem is different from the piercing set problem. Here, the aim is to minimize the total distances between successive piercing points $\Sigma_{i=1}^{n-1}d(p_i,p_{i+1})$, where $n$ is the number of input objects.
Friedman and Linial~\cite{FriedmanL93} proved that no deterministic online algorithm could achieve a competitive ratio better than $\sqrt{d}$ {for chasing convex objects in $\IR^d$}. They gave an online algorithm with a competitive ratio of at most $9\sqrt{2} + 5\sqrt{10}$ {for chasing lines in $\IR^2$}.
Later, Bienkowski et al.~\cite{BienkowskiBCCJK19} improved this result by presenting an algorithm with a competitive ratio of at most 3  for chasing lines in $\IR^d$.  
Bubeck et al.~\cite{BubeckLLS19} presented an online algorithm with a competitive ratio of $2^{O(d)}$ for {chasing convex objects in $\IR^d$}. 
Parallel to this, Bubeck et al.~\cite{BubeckKLLS20} presented an online algorithm having competitive ratio of $O(\min(d,\sqrt{d\log n}))$ when convex objects in $\IR^d$ are nested. Later, Sellke~\cite{Sellke20} and Argue et al.~\cite{ArgueGTG21} independently improved the result by showing that there exists an online algorithm having a competitive ratio of $O(\sqrt{d\log n})$ and $O(\min(d,\sqrt{d\log n}))$, respectively, for  chasing convex objects {in $\IR^d$}.

\subsection{Our Contributions}
First, we observe that the competitive ratio of every deterministic online algorithm for piercing intervals in $\mathbb{R}$ is at least $\Omega(n)$, where $n$ is the length of the input sequence. (Theorem~\ref{lower_bd_int}).  
Next, we show that piercing translated copies of a convex object is equivalent to the unit covering problem (Theorem~\ref{equivalence}). 
As an implication, all the results available in the literature~\cite{ChanZ09, CharikarCFM04, DumitrescuGT20, DumitrescuT22} for the online unit covering problem using translates of a convex object would  carry forward to the  piercing set problem for translates of a convex object. 
As a result, due to~\cite{DumitrescuGT20}, for translates of a centrally symmetric convex object $C$,  the competitive ratio of any deterministic online algorithm is at least the illumination number of the object $C$ (Theorem~\ref{illum}). 

Next, we ask what would happen when objects are {similarly sized}.
We consider similarly sized fat objects. For our purpose,  we define $\alpha$-fat and $\alpha$-aspect$_{\infty}$ fat objects in Section~\ref{bounded}. 
For $\alpha$-fat objects, the value of $\alpha$ is invariant under translation, rotation, reflection and scaling, whereas for $\alpha$-aspect$_{\infty}$ fat objects, the value of $\alpha$ is invariant under translation, reflection and scaling.  For any object in $\IR^d$, the value of $\alpha$ is  in the  interval $(0,1]$.

To obtain an upper bound, we consider a {well-studied}~{\cite{DumitrescuGT20,DumitrescuT22}} algorithm, $\AC$, that works as follows.  On receiving a new input object $\sigma$, if the existing piercing set does not pierce it, our online algorithm adds the center of $\sigma$ to the piercing set. 
For a hypercube $\sigma$, {a point in $\sigma$ is a \emph{center}} from which the maximum distance from any point of $\sigma$ is minimized. For $\alpha$-fat and $\alpha$-aspect$_{\infty}$ fat objects, a generalized definition of the center is given in Section~\ref{bounded}.  { We refer to Section~\ref{bounded} for the definition of the width of a fat object.}

\begin{enumerate}
 
\item First, we prove that for piercing similarly sized $\alpha$-aspect$_{\infty}$ fat objects in $\IR^d$, having width in the range $[1,k]$, \AC\ achieves a competitive ratio  of at most~${\left({\Bigl\lceil}2\left(1+\frac{1}{\alpha}\right)\Bigr\rceil^d-\Bigl\lfloor\frac{2}{\alpha}\Bigr\rfloor^d\right)} {\Bigl\lceil}\log_{1+\alpha}(\frac{2k}{\alpha}){\Bigr\rceil} $ $+1$(Theorem~\ref{thm:alpha_inf}).

\item  Next, we show that for piercing similarly sized $\alpha$-fat objects in $\IR^3$, having width in the range  $[1,k]$, \AC\ achieves  a competitive ratio of at most~${\left(\left(1+\frac{1}{\sin(\theta/2)}\right)^3-1\right)\Bigl\lceil}\log_{1+x}(\frac{2k}{\alpha}){\Bigr\rceil}+1$, where $\theta=\frac{1}{2}\cos^{-1}\left(\frac{1}{2} +\frac{1}{1+\sqrt{1+4\alpha^{2}}}\right)$ and  $x=\frac{\sqrt{1+4{\alpha}^2}-1}{2}$ (Theorem~\ref{thm:alpha-ball}). 
We achieve a similar result for piercing similarly sized $\alpha$-fat objects in $\IR^2$ having width in the range  $[1,k]$ (Theorem~\ref{thm10}).

\item Then, for piercing homothetic hypercubes in $\IR^d$, having side length in the range  $[1,k]$, we propose an algorithm,  $\AV$, that achieves a competitive ratio of at most~$3^d{\lceil}\log_2 k{\rceil}+2^d$ (Theorem~\ref{thm:d-hyp}).
\item Next, we consider the {lower bound for the problem}. We prove that the competitive ratio of every deterministic online algorithm  for piercing similarly sized $\alpha$-fat objects in $\IR^2$, {having width in the range $[1,k]$}, is at least~${\Bigl\lfloor}\log_{\frac{2+{\epsilon}}{\alpha}}(k){\Bigr\rfloor} + 1$,  where $0<\epsilon<\frac{1}{4}$ is a sufficiently small constant close to $0$ (Theorem~\ref{thm:lowerBoundPolygon}).

\item Later, generalizing this lower bound result to higher dimensions, we show that the competitive ratio of every deterministic online algorithm for piercing homothetic hypercubes  in $\IR^d$, having side length in the range $[1,k]$, is at least $d{\Bigl\lfloor}\log_{(2+\epsilon){^2}}k{\Bigr\rfloor}+2^d$,  where $0<\epsilon<\frac{1}{4}$ is a sufficiently small constant close to $0$ (Theorem~\ref{thm_sq_d}).
\end{enumerate}

 \subsection{Organization} 
 The paper is organized as follows. First, we give some relevant definitions and preliminaries in 
 Section~\ref{Notations}.  The relationship between the unit covering and the unit piercing problem is discussed in Section~\ref{Unit_cover}.
 Definitions related to the fat object are given in Section~\ref{bounded}.
We present all the upper bound related results in Section~\ref{upp}. All the lower bound related results are  in Section~\ref{sec:lb}. Finally, in Section~\ref{Conclusion}, we give a conclusion.


\section{Preliminaries}\label{Notations}
{We use $\mathbb{Z}^{+}$ and $\IR^{+}$ to denote the set of positive integers and positive real numbers.}
{For any $n\in\mathbb{Z}^{+}$}, we use $[n]$ to denote the set $\{1,2,\ldots,n\}$.  {For any $i\in[d]$, the $i$th coordinate of a point $p\in\IR^d$ is denoted  by $p(x_i)$.}
If not explicitly mentioned, we use the term \emph{object} to denote a compact set in $\IR^d$ with a nonempty interior.
The interior of an object $\sigma$ is denoted by $\interior(\sigma)$.
{A set $\cal S$ of objects is said to be \emph{similarly sized} if the ratio of the largest diameter of an object in $\cal S$ to the smallest diameter of an object in $\cal S$ is bounded by a fixed constant.}
Any set of the form $\lambda \sigma+x=\{\lambda c+x\ |\ c\in \sigma\}$, where $x\in\IR^d$ and $\lambda \in {\IR}^+$, is called a \emph{homothetic copy} of $\sigma$. A set $\cal S$ of objects is said to be \emph{homothetic} if each object $\sigma\in \cal S$ is a homothetic copy of every other object $\sigma' \in \cal S$. The \emph{illumination number} of an object $\sigma$, denoted by $I(\sigma)$, is the minimum number of smaller homothetic copies ($\lambda<1$) of $\sigma$ whose union contains $\sigma$~\cite{Pach}.
A point $y'$ is a \emph{reflection of a point} $y$ in a plane over a point $x$ if the midpoint of the line segment $\overline{yy'}$ is $x$.
Now, we define the \emph{reflection of an object} $\sigma$ over a point $c\in \IR^d$, denoted by $-\sigma(c)$, as $-\sigma(c)=\{y\in\IR^d\ |\ \frac{x+y}{2}=c \text{ and } x\in\sigma\}=\{2c-x\ |\ x\in\sigma\}$(see Figure~\ref{fig:ref_1}).
Note that for any $c_1,c_2\in\IR^d$, $-\sigma(c_1)=-\sigma(c_2)+\tau$, where $\tau=2(c_1-c_2)$ is a translating vector. We use $-\sigma$ to denote a translated copy of $-\sigma(c)$.
An object is considered  \emph{centrally symmetric} if there exists a point $c\in\sigma$ such that $\sigma=-\sigma(c)$.

Let $x$ and $y$ be any two points in $\IR^d$. We use $d(x,y)$ (respectively, $d_{\infty}(x,y)$) to denote the distance between  $x$ and $y$ under the $L_2$ (respectively, $L_{\infty}$) norm.  
Let $C$ be a convex object containing the origin in its interior.
 We translate $C$ by vector $x$ and consider the ray from $x$ through $y$. Let $v$ denote
the unique point on the boundary of $C$ hit by this ray.
The function $d_{C}(x,y)$, induced by the object $C$,
is defined as $d_{C}(x,y)=\frac{d(x,y)}{d(x,v)}$~\cite{IckingKLM95}. 
Note that the function $d_{C}$ is a metric (distance function) when $C$ is a centrally symmetric convex object, and in general, the following property is well-known. 
\begin{property}\label{prop:conv}
    Let $C$ be any convex object and  $x, y$ be any two points in $\IR^d$. Then, $d_C(x,y)=d_{-C}(y,x)$. 
\end{property}
 \begin{figure}[htbp]
    \centering
    \includegraphics[width=40 mm]{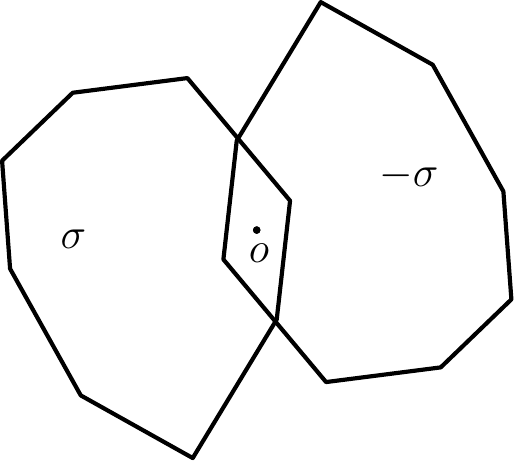}
    \caption{Reflection of an object $\sigma$ through the point $o$.}
    \label{fig:ref_1}
\end{figure}

\subsection{Lower Bound}\label{Lower}
First, we observe that the competitive ratio of the piercing set problem has a pessimistic lower bound of $\Omega(n)$, for one-dimensional intervals.
\begin{theorem}\label{lower_bd_int}
The competitive ratio of every deterministic online algorithm for piercing intervals in $\IR$ is at least~$\Omega(n)$, where $n$ is the length of the input sequence.
\end{theorem}
\begin{proof}
Let $\sigma_1$ be the first interval presented by the adversary to the online algorithm. 
Let $p_1$ be a point placed by the online algorithm 
to pierce the interval $\sigma_1$. The point $p_1$ partitions the interval $\sigma_1$ into two parts, of which, let $\sigma_1^{L}$ be a larger part that does not contain the point $p_1$. Now, the adversary can place an interval $\sigma_2$ completely contained in  $\sigma_1^{L}$. 
For the new interval $\sigma_2$, any online algorithm needs a new piercing point $p_2$.
Now again, one can  define a partition $\sigma_2^{L}$ of $\sigma_2$ depending on the position of the point $p_2$ such that $\sigma_2^{L}$  does not contain $p_2$, and the adversary will place an interval $\sigma_3$ completely contained in $\sigma_2^{L}$. In this way, the adversary can adaptively construct $n$ intervals for which any online algorithm needs $n$ distinct points to pierce, while an offline optimum needs only one point. Hence, the lower bound of the competitive ratio is $\Omega(n)$.
\end{proof}

\subsection{Unit Covering vs Unit Piercing }\label{Unit_cover}
\begin{definition}[Unit Piercing Problem]
For any $d\in\mathbb{Z}^{+}$, given a family $\s$ of translated copies of a convex object $C\subset\IR^d$, in the unit piercing problem, we need to pierce each  object in $\s$ by placing the minimum number of points in $\IR^d$.
\end{definition}

\begin{definition}[Unit Covering Problem]
For any $d\in\mathbb{Z}^{+}$, given a set of points $\P \subseteq \IR^d$, in the unit covering problem, we need to place the minimum number of translated copies of a convex object $-C\subset\IR^d$ to cover all the points in $\P$.
\end{definition}
The following theorem connects the above two problems.
\begin{theorem}\label{equivalence}
The unit piercing problem is equivalent to the unit covering problem.
\end{theorem}
To prove this, first, we demonstrate the following lemma.
\begin{lemma}\label{lemma_d_C}
Let $x$ and $y$ be any two points in $\IR^d$. Then $x$ lies in $y+C$ if and only if $y$ lies in $x+(-C)$.
\end{lemma}
\begin{proof}
If  $x$ lies in $y+C$, then $d_C(y,x) \leq 1$. Due to Property~\ref{prop:conv}, $d_{-C}(x,y) =d_C(y,x)\leq 1$. Therefore, $y$ lies in $x+(-C)$. Similarly, one can prove the converse. Hence, the lemma follows.
{An alternative geometry-based proof is as follows.
Let $x$ lies in $y+C$.
Note that $x$ lies in $y+C$ if and only if $x-y$ lies in $C$ if and only if $y-x$ lies in $-C$ if and only if $y$ lies in $x+(-C)$. The first and third equivalences are easy to follow; the second equivalence follows by reflecting on the origin.}
\end{proof}

\noindent
\textit{Proof of Theorem~\ref{equivalence}.}
For the unit piercing problem, let us assume that each unit object in $\cal S$ is a translated copy of $C\subset \IR^d$.  Due to Lemma~\ref{lemma_d_C}, we know that some point $x\in \IR^d$ pierces an object $y+C$ if and only if the object $x+(-C)$ covers the point $y$. Thus,  we can convert this problem to an equivalent unit covering problem, where the set $\P$ of points are the center of all the objects in $\cal S$, and we need to cover them by translates of $(-C)$. 
In a similar fashion, one can prove the other side of the equivalence. \myqed

As a consequence of Theorem~\ref{equivalence}, all the results related to the online unit covering problem
(summarized in Table~\ref{table_contri}) studied in~\cite{ChanZ09, CharikarCFM04, DumitrescuGT20, DumitrescuT22} are carried for the online piercing problem.
Specifically, due to~\cite[Theorem~4]{DumitrescuGT20}, we have the following lower bound result for centrally symmetric convex objects.  A proof adapted from~[ibid.], which might be of independent interest, is given for the sake of completeness.
\begin{theorem}\label{illum}
The competitive ratio of every deterministic online algorithm for piercing a set of translated copies of a centrally symmetric convex object $C$ is at least $I(C)$, where $I(C)$ denotes the illumination number of $C$.
\end{theorem}
\begin{figure}[htbp]
  \centering
     \begin{subfigure}[b]{0.45\textwidth}
          \centering
        \includegraphics[width=60mm]{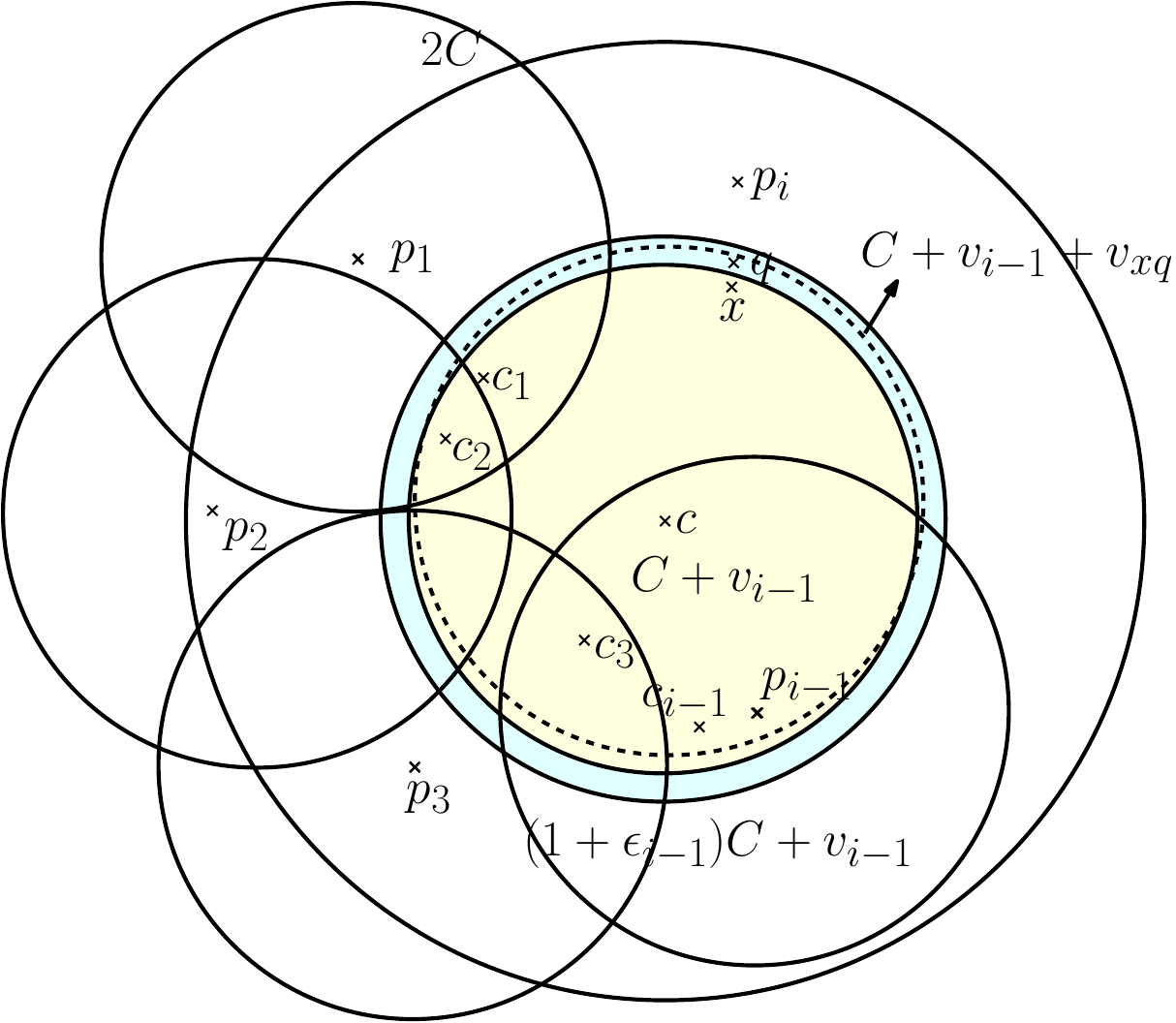}
    \caption{}
    \label{fig:illum_piercing}
     \end{subfigure}
      \hfill
       \begin{subfigure}[b]{0.45\textwidth}
    \centering
    \includegraphics[width= 60 mm]{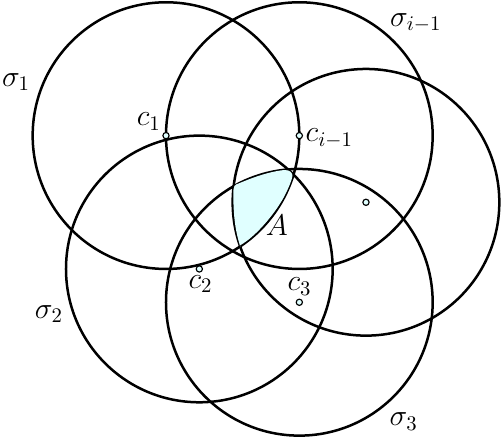}
    \caption{}
    \label{fig:intersecting}
\end{subfigure}
 \caption{(a) A centrally symmetric convex object $C+v_{i-1}$ (colored yellow) contains points $c_1,c_2,\ldots,c_{i-1}$. Note that $(1 + \epsilon_{i-1})C + v_{i-1}$ is $\epsilon_{i-1}$-neighbourhood of $C+v_{i-1}$. The distance $d_C(q,x)<\epsilon_{i-1}$, where $q \in \interior((1 + \epsilon_{i-1})C + v_{i-1})$ and $x \in \interior(C + v_{i-1})$. The translate $C+v_{i-1}+v_{xq}$ (in dotted lines) contains $c_1,c_2,\ldots,c_{i-1}$ and $q$. (b) Cyan colored region, denotes the intersection region $A$, i.e., $\cap_{k=1}^{i-1} \sigma_k$.}
\end{figure}

\begin{proof}
To prove the lower bound, we can think of a game between Alice and Bob. Here, Alice plays the role of the adversary, and Bob plays the part of the online algorithm. Alice presents a translated copy $\sigma_i$ of $C$ to Bob, and  Bob needs to place a piercing point $p_i$ to pierce $\sigma_i$ if some previously placed piercing points do not already pierce it, for $i\in [I(C)]$.
We claim that Alice can find an object $\sigma_i$ centered at $c_i$ such that the following invariants hold.
\begin{itemize}
    \item[(I)] The object $\sigma_i$ is not pierced by any of the previously placed piercing point $p_j$, where {$j\in[i-1]$}.
    \item[(II)] The set of objects $\{\sigma_1,\sigma_2,\ldots,\sigma_i\}$ can be pierced by a single point. In other words, $\bigcap^i_{k=1}\sigma_k\neq \emptyset$.
\end{itemize}
Invariant (I) implies that Bob is forced to place $I(C)$ many piercing points, while invariant (II) ensures that the set of objects $\{\sigma_1,\sigma_2,\ldots,\sigma_i\}$ can be pierced by a single point. Consequently, an offline optimum needs just one piercing point, while any online algorithm needs $I(C)$ many piercing points to pierce $\sigma_1,\sigma_2,\ldots,\sigma_{I(C)}$. Hence, the competitive ratio of any deterministic online algorithm is at least $I(C)$.
We prove this claim by induction {on $i$}.  For $i=1$, both invariants (I) and (II) trivially hold (Base case).
For $i\in\{1,2,\ldots I(C)-1\}$, assume that both invariants  hold (Induction hypothesis).
Now it is enough to show that Alice can find an object $\sigma_{i}$ centered at $c_{i}$ such that both invariants hold for $i=I(C)$.
\begin{clm}\label{claim:Conatin}
Let $\sigma_1,\sigma_2,\ldots, \sigma_{i-1}$ be translated copies of a centrally symmetric convex object $C$, centered at $c_1,c_2, \ldots,c_{i-1}$, respectively. Here, $\cap_{j=1}^{i-1} \sigma_j$ is nonempty if and only if there exists a translating vector $v_{i-1}$ such that $C+v_{i-1}$  contains all the centers $c_1,\ldots, c_{i-1}$.
\end{clm}
\begin{proof}
Let us assume that $\cap_{j=1}^{i-1} \sigma_j$ is nonempty, and let $x$ be any point in $\cap_{j=1}^{i-1} \sigma_j$ (see Figure~\ref{fig:intersecting}). For each $j\in[i-1]$, the $d_C$ distance of $x$ to all  centers $c_j$ is less than one, i.e., $d_C(x,c_j)< 1$. Therefore, a translate of $C$, centered at $x$, will contain all the centers $c_1,c_2,\ldots,c_{i-1}$ of the objects $\sigma_1,\sigma_2,\ldots,\sigma_{i-1}$, respectively.
For the converse, let us assume that a translate of $C$ (say $C+v_{i-1}$), centered at $c$, contains all the centers $c_1,\ldots, c_{i-1}$.  Here, for each $j\in [i-1]$, $d_C(c,c_j)$ is less than one. Since $C$ is a centrally symmetric object, for each $j\in [i-1]$, $d_C(c_j,c)$ is also less than one. As a result, the point $c$ lies in $\sigma_j$, for each $j\in [i-1]$.  Therefore, $\cap_{j=1}^{i-1} \sigma_j$ is nonempty.
\end{proof}
 Let $C+v_{i-1}$ be a translate of $C$ containing all the centers $c_1,c_2,\ldots,c_{i-1}$.
 Let $\epsilon_{i-1}$ be the largest value such that translates of $\epsilon_{i-1}C$ centered at $c_1,c_2,\ldots,c_{i-1}$ are contained in $C + v_{i-1}$, where $v_{i-1}\in \IR^d$ is a translation vector (see Figure~\ref{fig:illum_piercing}).
Note that $(1+\epsilon_{i-1})C+v_{i-1}$ is the $\epsilon_{i-1}$-neighbourhood of $C+v_{i-1}$ in the $L_C$ norm.

Since $p_j$ is a piercing point for $\sigma_j$, the $d_C$ distance between the center $c_j$ and the piercing point $p_j$ for the object $\sigma_j$ is at most one, i.e., $d_C(c_j,p_j)\leq 1$, for each $j\in[i-1]$. Since $d_C(p_j,c_j)\leq1$, so all the piercing points $p_1,p_2,\ldots,p_{i-1}$ are going to lie in the $C$-neighbourhood of $C$ (see Figure~\ref{fig:illum_piercing}).
\begin{clm}\label{cliam:surety}
Let $P_1,P_2,\ldots,P_{i-1}$ be translates of $C$, centered at $p_1,p_2,\ldots,p_{i-1}$, respectively. Then $\cup_{j=1}^{i-1} P_j$ will not entirely cover $(1+\epsilon_{i-1})C+v_{i-1}$.
\end{clm}
\begin{proof}
If  $\cup_{j=1}^{i-1} P_j$  covers the entire region $(1+\epsilon_{i-1})C+v_{i-1}$, then it contradicts the  definition of illumination number. 
\end{proof}

\noindent
Due to Claim~\ref{cliam:surety}, we can place a point $c_{i}$ (center of the object $\sigma_{i}$) in $\interior((1+\epsilon_{i-1})C+v_{i-1})$ such that $d_C(p_j,c_{i})>1$, for all $j\in[i-1]$.
Now, if we can show that some translate of $C$  covers $c_{i}$ along with all the previous $c_j$'s, where $j \in [i-1]$,  then we are done.
\begin{clm}\label{claim:new_q}
For any point $q \in (1+\epsilon_{i-1})C+v_{i-1}$, there exists some translate of $C$ $($say $\sigma_{new})$ such that the interior of $\sigma_{new}$ contains all centers $\{c_1,\dots,c_{i-1},q\}$.
\end{clm}
\begin{proof}
 If $q \in \interior(C+v_{i-1})$ (colored yellow, see Figure~\ref{fig:illum_piercing}), then $C+v_{i-1}$ contains all the previous centers $c_1,\ldots, c_{i-1}$ as well as $q$. 
    Without loss of generality, assume that $q$ lies in the annular region $\interior((1+\epsilon_{i-1})C+v_{i-1})\setminus \interior(C+v_{i-1})$ (colored blue, see Figure~\ref{fig:illum_piercing}).
It is easy to see that there exists a point $x\in \interior(C+v_{i-1})$ such that the vector $v_{xq}=q-x$ is less than $\epsilon_{i-1}$ (under the $L_C$ norm). Note that $q\in \interior(C+v_{i-1}+v_{xq})$ (see Figure~\ref{fig:illum_piercing}).
 
By choice of $\epsilon_{i-1} > 0$, the $d_C$ distance between the points
$c_1,c_2,\ldots,c_{i-1}$ and $\partial(C + v_{i-1})$ is at least $\epsilon_{i-1}$, and $v_{xq}$  is also less than $\epsilon_{i-1}$. It implies that the $\interior(C+v_{i-1}+v_{xq})$ contains the centers $c_1,c_2,\ldots,c_{i-1}$ as well as $q$. 
\end{proof}
By the above Claim~\ref{claim:new_q}, one can find a translate of $C$ such that it covers the point $c_{i}$ along with all the previous points $c_j$, where $j \in [i-1]$. 
From Claim~\ref{claim:Conatin}, $\cap_{j=1}^{i} \sigma_j$ is also nonempty. Hence, the theorem follows.
\end{proof}

\noindent

It is known that $I(C) =2^d$ for any full-dimensional parallelepiped in $\IR^d$~\cite{DumitrescuGT20,Pach}. Consequently, the value of $I(C)=2^d$, for hypercubes in $\IR^d$. Therefore, we have the following.

\begin{corollary}\label{corr_illum}
The competitive ratio of every deterministic online algorithm for piercing translates of a hypercube in $\IR^d$ is at least $2^d$.
\end{corollary}

\begin{table}[htbp]
\centering
\begin{tabular}{|p{6 cm}|p{3cm}|p{3cm}|}
\hline 
  Geometric Objects &  Lower Bound & Upper Bound \\
\hline
\hline
Translated copies of an interval & $2$~\cite{CharikarCFM04} & $2$~\cite{ChanZ09}\\
 \hline
Translated copies of a  square &  $4$~\cite{DumitrescuGT20, DumitrescuT22} & $4$~\cite{ChanZ09} \\
  \hline
Translated copies of a  hypercube in $\IR^d$ &  $2^d$~\cite{ DumitrescuT22} & $2^d$~\cite{ChanZ09}\\
\hline
Congruent disks & $4$~\cite{DumitrescuGT20} & $5$~\cite{DumitrescuGT20} \\
\hline
Congruent balls in $\IR^3$ & $5$~\cite{DumitrescuGT20} & $12$~\cite{DumitrescuGT20} \\
\hline
Congruent balls in $\IR^d$, $d>3$ & $d+1$~\cite{DumitrescuGT20} & $O({1.321}^d)$~\cite{DumitrescuGT20} \\
\hline
  Translated copies of a centrally symmetric convex object $C\subset \IR^d$  & $I(C)$~\cite{DumitrescuGT20} & 
  $\bigstar$~[Corollary~\ref{lem:unit}]\\
\hline
\end{tabular}
  \begin{flushleft}
  \item \small{$\bigstar$ \text{Result obtained in this paper.}}
 \end{flushleft}
\caption{Summary of known results for the online piercing set problem.}
\label{table_contri}
\end{table}
\normalsize

\subsection{Fat Object}\label{bounded}
A number of different definitions of fatness (not extremely long and skinny) are available in the geometry literature~\cite{Chan03}. For example, see many references in~\cite{BergSVK02, EfratKNS00}. For our purpose, we give the following appropriate definition.

Let $\sigma$ be an object and $x$ be any point in $\sigma$. Let $\alpha(x)$ be the ratio between the minimum and maximum distance (under the $L_2$ norm) from $x$ to the boundary $\delta(\sigma)$ of the object $\sigma$. In other words, $\alpha(x)= \frac{\min_{y\in \delta(\sigma)}d(x,y)}{\max_{y\in \delta(\sigma)}d(x,y)}$, where $d(x,y)$ is the distance between $x$ and $y$.
The \emph{aspect ratio} $\alpha(\sigma)$ of an object $\sigma$ is defined as the maximum value  of $\alpha(x)$ for any point  $x\in \sigma$, i.e., $\alpha(\sigma)=\max \{\alpha(x)  : x \in  \sigma\}$. An object is said to be an \emph{$\alpha$-fat object} if its aspect ratio is exactly $\alpha$. 
A point $c\in \sigma$ with $\alpha(c)=\alpha(\sigma)$ is defined as a \emph{center} of the object $\sigma$ (see Figure~\ref{fig:aspect_2}).  Note that the center of an $\alpha$-fat object might not be unique. For an example, see Figure~\ref{fig:not-unique}. If an object $\sigma$ has multiple points satisfying the property of a center, then for our purpose, we can arbitrarily choose any one of them as the center of $\sigma$. The minimum (respectively, maximum) distance from  the center to the boundary of the object is referred to as the \emph{width} (respectively, \emph{height}) of the object. 

Note that for any object $\sigma$, we have $0<\alpha(\sigma)\leq 1$. The maximum possible value of $\alpha$ is attained when the object is a ball in $\IR^d$, while the value of $\alpha$ is $\frac{1}{\sqrt{d}}$ when the object is a hypercube in $\IR^d$. We say that a set $\cal S$  of objects is \emph{fat} if there exists a constant $0<\alpha\leq 1$ such that each object in $\cal S$ is $\alpha$-fat. For a set $\cal S$ of fat objects,  each object $\sigma\in {\cal S}$ need not be convex,  and it does not need to be connected.
The set $\cal S$ of $\alpha$-fat objects is said to be \emph{similarly sized fat objects} when the ratio of the largest width of an object in $\cal S$ to the smallest width of an object in $\cal S$  is bounded by a fixed constant.
\begin{figure}[htbp]
  \centering
  \begin{subfigure}[b]{0.33\textwidth}
          \centering
        \includegraphics[width=50 mm]{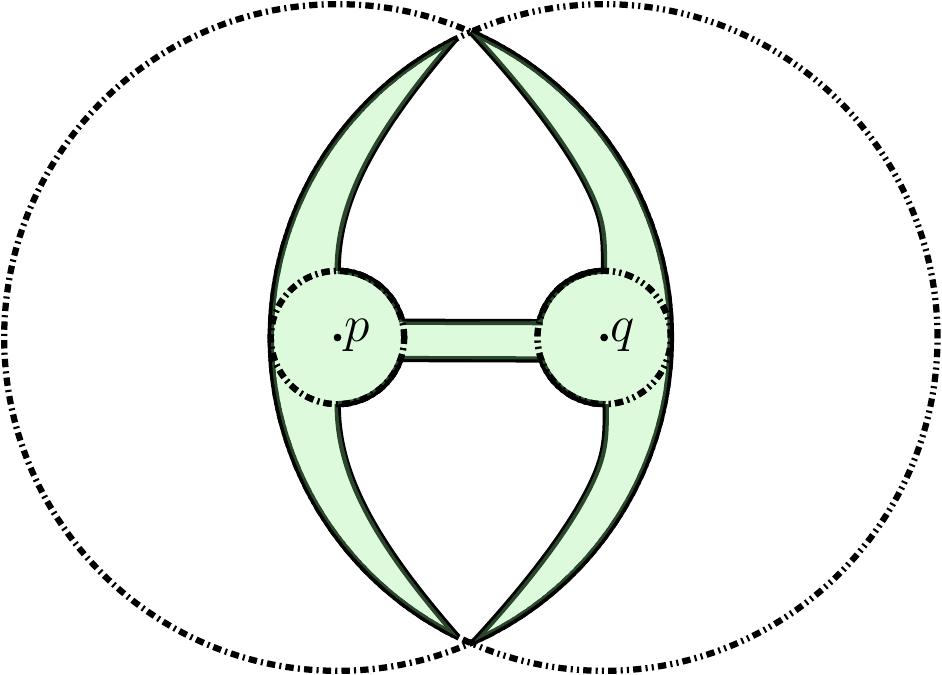}
    \caption{}
    \label{fig:not-unique}
     \end{subfigure}
     \hfill
    \begin{subfigure}[b]{0.32\textwidth}
          \centering
        \includegraphics[width=35mm]{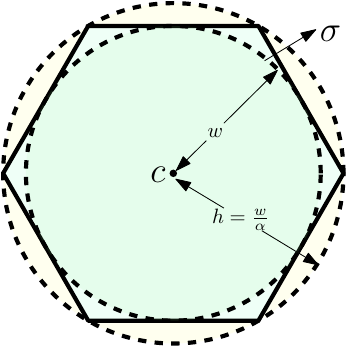}
    \caption{}
    \label{fig:aspect_2}
     \end{subfigure}
     \hfill
      \begin{subfigure}[b]{0.33\textwidth}
          \centering
    \includegraphics[width=35mm]{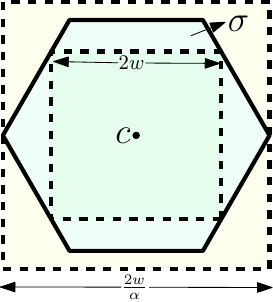}
    \caption{ }
    \label{fig:inf}
     \end{subfigure}
 \caption{(a)
 The object colored green has two centers at $p$ and $q$. (b-c) Geometric interpretation of aspect ratio and aspect$_{\infty}$ ratio, respectively.}
\end{figure}

Observe the following properties of an $\alpha$-fat object.
\begin{property}\label{prop1}
Let $\sigma\subset \IR^d$ be an $\alpha$-fat object, then 
\begin{itemize}
    \item any rotated copy $\sigma'$ of $\sigma$ is  an $\alpha$-fat object;
    \item any translated copy $\sigma'$ of $\sigma$ is  an $\alpha$-fat object;
    \item any reflected copy $-\sigma$ of $\sigma$ is  an $\alpha$-fat object;
    \item the scaled copy $\lambda \sigma$ is  an $\alpha$-fat object, where $\lambda\in \IR^{+}$.
\end{itemize}
\end{property}

\begin{property} \label{prop1.2}
Let $\sigma\subset \IR^d$ be an $\alpha$-fat object  centered at a point $c$ with width $w$, then

\begin{itemize}
\item a ball  of radius $w$ centered at $c$  is completely contained in $\sigma$, and
\item a ball  of radius $\frac{w}{\alpha}$ centered at $c$  contains the object $\sigma$. 
\end{itemize}
\end{property}

Considering the $L_{\infty}$ norm instead of the $L_2$ norm, similar to the above,
one can  define the aspect$_\infty$ ratio, center, width and  height  of an object. An object with aspect$_\infty$ ratio $\alpha$ is said to be an 
 \emph{ $\alpha$-aspect${_\infty}$  fat object} (see Figure~\ref{fig:inf}). 
  Note that for any object $\sigma$, the value of $\alpha_{\infty}(\sigma)$ is also strictly greater than zero  and the maximum possible value of $\alpha_{\infty}(\sigma)$ is attained for an axis-aligned hypercube. Analogous to similarly sized fat objects, we can define similarly sized aspect$_{\infty}$ fat objects. Similar to properties~\ref{prop1} and~\ref{prop1.2}, we have the following properties.

\begin{property}\label{prop2}
Let $\sigma\subset \IR^d$ be an $\alpha$-aspect${_\infty}$ fat object, then 
\begin{itemize}
    \item any translated copy $\sigma'$ of $\sigma$ is  an $\alpha$-aspect${_\infty}$ fat object;
    \item  any reflected copy $-\sigma$ of $\sigma$ is  an $\alpha$-aspect${_\infty}$ fat object;
    \item the scaled copy $\lambda \sigma$ is  an $\alpha$-aspect${_\infty}$ fat object, where $\lambda\in \IR^{+}$.
\end{itemize}
\end{property}

\begin{property} \label{prop2.2}
Let $\sigma\subset \IR^d$ be an $\alpha$-aspect${_\infty}$ fat object object centered at a point~$c$ with width~$w$, then

\begin{itemize}
\item a hypercube  of side length $2w$ centered at $c$  is completely contained in $\sigma$, and
\item a hypercube  of side length $\frac{2w}{\alpha}$ centered at $c$  contains the object $\sigma$. 
\end{itemize}
\end{property}

\section{Upper Bound}\label{upp}

In subsections~\ref{upp:inf} and \ref{ball-3-d}, we analyse the performance of $\AC$ for  piercing similarly sized $\alpha$-aspect$_{\infty}$ fat objects and $\alpha$-fat objects, respectively,  having width in the range  $[1,k]$. The analysis is similar in nature for both. To bound the competitive ratio, we determine the number of  points placed by $\AC$ against each point $p$ in an offline optimum. To compute the number of piercing points placed by our algorithm, we consider the region containing the centers of all objects that can be pierced by the point $p$. We have partitioned this region into $\lceil\log k\rceil$ annular regions such that $\AC$ places the same number of piercing points in each of these annular regions.  Finally, we give an upper bound on the total number of points placed by our algorithm in each of these annular regions. The competitive ratio is $\lceil\log k\rceil$  multiplied by this number.

\subsection{ For  {Similarly Sized} Aspect${_{\infty}}$ Fat Objects}\label{upp:inf}

This section presents an upper bound of the competitive ratio for piercing  similarly sized $\alpha$-aspect${_{\infty}}$ fat  objects in $\IR^d$.  In the proof of the following theorem, all the distances, if explicitly not mentioned, are under the $L_{\infty}$ norm, and all the hypercubes  are axis-parallel. 
\begin{theorem} \label{thm:alpha_inf}
For piercing similarly sized $\alpha$-aspect${_\infty}$ fat objects in $\IR^d$ having  width in the range  $[1,k]$, $\AC$ achieves a competitive ratio of at most~${\left(\Bigl\lceil2\left(1+\frac{1}{\alpha}\right)\Bigr\rceil^d-\Bigl\lfloor\frac{2}{\alpha}\Bigr\rfloor^d\right)} {\Bigl\lceil}\log_{1+\alpha}(\frac{2k}{\alpha}){\Bigr\rceil} +1$.
\end{theorem}
\begin{figure}[htbp]
    \centering
\begin{subfigure}[b]{0.48\textwidth}
          \centering
        \includegraphics[width=43 mm]{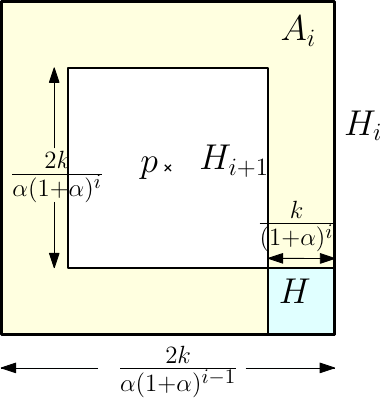}
    \caption{}
    \label{fig:sq_poly}
     \end{subfigure}
      \hfill
      \begin{subfigure}[b]{0.48\textwidth}
    \centering
    \includegraphics[width= 65mm]{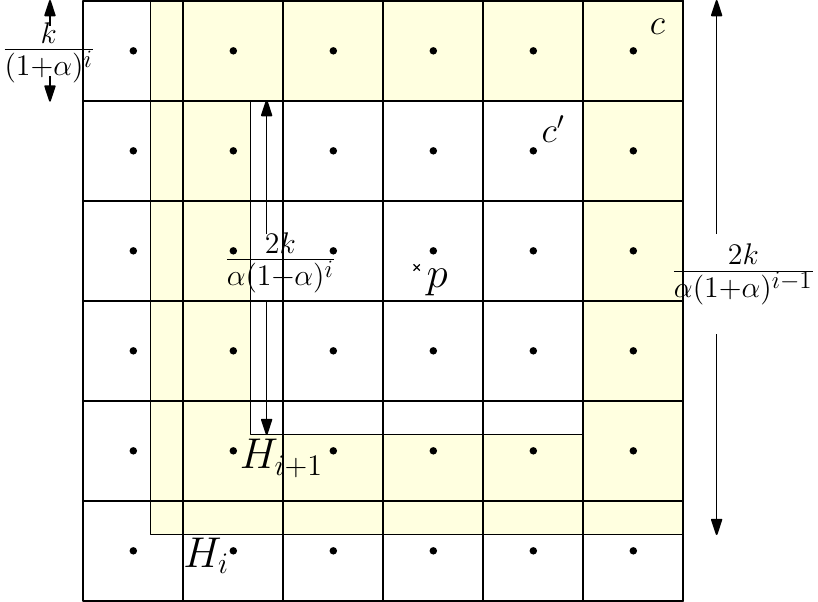}
    \caption{}
    \label{fig:tiling}
\end{subfigure}
    \caption{(a) Illustration of annular region $A_{i}=H_{i}\setminus H_{i+1}$. (b) {The cells of $\Delta'$ are depicted, where one of the corners of a cell $c\in\Delta'$ coincides with one of the corners of $H_i$}.}
    \label{fig:sq1}
\end{figure}
\begin{proof}
 Let  $\I$ be the set of $\alpha$-aspect${_\infty}$ fat objects presented to the algorithm.  Let $\A$ and $\OO$ be the piercing set returned by $\AC$ and an offline optimal for $\I$.   
Let $p \in \OO$ be a piercing point and let $\I_p\subseteq \I$ be the collection of all input $\alpha$-aspect${_\infty}$ fat objects pierced by the point $p$.
 Let $\A_{p}\subseteq \A$ be the set of piercing points placed by \AC\ to pierce all the input  objects in $\I_p$. 
 It is easy to see that $\A=\cup_{p\in \OO}\A_p$.
Therefore, the competitive ratio of our algorithm is  upper bounded by  $\max_{p\in\OO}|\A_p|$. 

Let us consider any point $a\in \A_p$. Since $a$ is the center of an $\alpha$-aspect${_\infty}$ fat object $\sigma\in \I_p$ containing the point $p$ and having width at most $k$, the distance  between $a$ and $p$ is at most $\frac{k}{\alpha}$.
 Therefore, a hypercube $H_1$ of side length $\frac{2k}{\alpha}$, centered at $p$, contains all the points in $\A_p$.  Let $H_i$ be a hypercube centered at $p$ having side length of ${\frac{2k}{\alpha(1+\alpha)^{i-1}}}$, where $i\in [m]$ and  $m$ is the smallest integer such that $\frac{2k}{\alpha(1+\alpha)^{m-1}}\leq 1$. Note that  $H_1, H_2,\ldots, H_m$ are concentric hypercubes. Let us define the annular region $A_{i}=H_{i}\setminus H_{i+1}$, where $i\in [m-1]$ {(see Figure~\ref{fig:sq_poly})}. Let $\A_{p,m}=\A_p\cap H_{m}$, and  $\A_{p,i}=\A_p\cap A_i$ be the subset of $\A_p$ that is contained in the region $A_i$, for $i\in [m-1]$. Since the distance  between any two points in $H_m$ is at most one and the width of any object in $\I_p$ is at least one, any object belonging to $\I_p$ having center in $H_m$ contains the entire hypercube $H_m$. As a result, our online algorithm places at most one piercing point in $H_m$. Thus, $|\A_{p,m}|\leq 1$.
\begin{figure}[htbp]
    \centering
\includegraphics[width=125 mm]{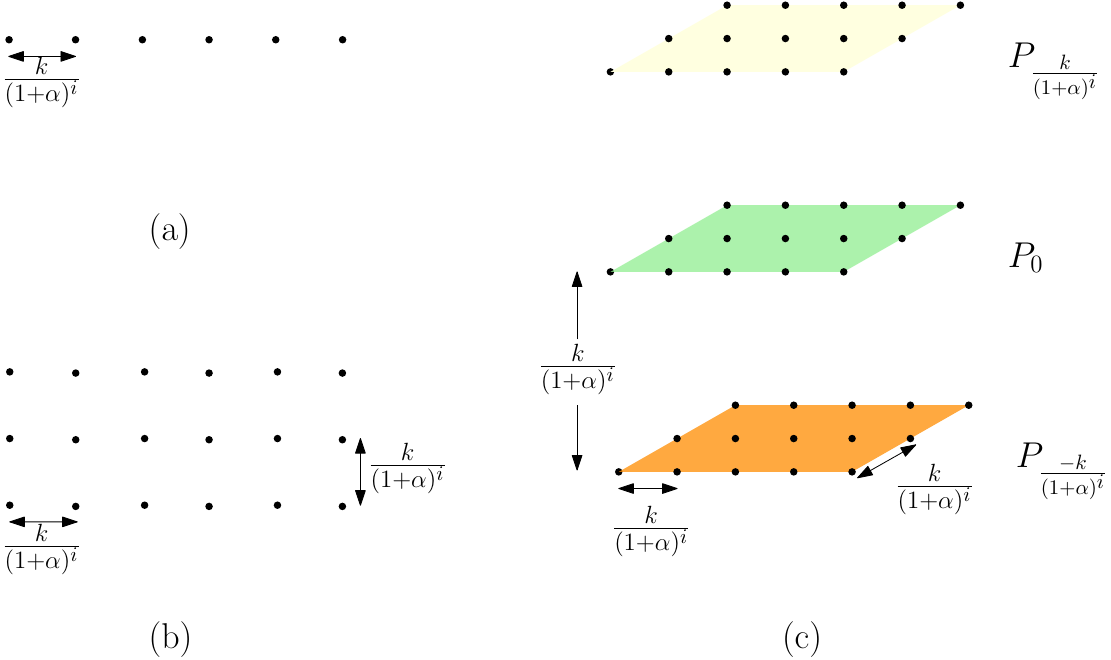}
       \caption{{Points of $\Pi_d$ are drawn  for (a) $d=1$, (b) $d=2$, and (c) $d=3$.  Let $P_q$ be the hyperplane $\{x\in\IR^{3}\ |\ x(x_{3})=q\}$. In Figure (c), the projection of $\Pi_3$ over planes $P_{\frac{k}{(1+\alpha)^i}}$ (colored yellow), $P_{0}$ (colored green) and $P_{\frac{-k}{(1+\alpha)^i}}$ (colored orange) over a rectangular region is depicted.}}
       \label{fig:lattice}
\end{figure}

\begin{lemma}\label{lm:Cardinality_A_i}
$|\A_{p,i}|\leq{\Bigl\lceil2\left(1+\frac{1}{\alpha}\right)\Bigr\rceil^d-\Bigl\lfloor\frac{2}{\alpha}\Bigr\rfloor^d},\text{ where }i\in [m-1]$.
\end{lemma}
\begin{proof}
Let $H$ be any hypercube of side length $\frac{k}{(1+\alpha)^i}$ {such that $H\cap A_i\neq\emptyset$}.
First, we argue that our online algorithm places at most one piercing point in $H\cap A_i$ to pierce the objects in $\I_p$. 
Let $q_1$ be the first  piercing point placed by our online algorithm in {$H\cap A_i$}. For a contradiction, let us assume that our online algorithm places another piercing point $q_2\in {H\cap A_i}$. Let $\sigma \in \I_p$ be the object centering at $q_2$ for which this piercing point was placed.  Since the object $\sigma$ contains $p$ and   the distance  between $p$ and $q_2$ is at least $\frac{k}{\alpha(1+\alpha)^i}$, the height of the object $\sigma$ is at least $\frac{k}{\alpha(1+\alpha)^i}$ and the width is at least $\frac{k}{(1+\alpha)^{i}}$. 
Note that the distance between any two points in $H$ is at most $\frac{k}{(1+\alpha)^i}$. Therefore, the distance between $q_1$ and $q_2$ is at most $\frac{k}{(1+\alpha)^i}$. Since the width of $\sigma$ is at least $\frac{k}{(1+\alpha)^{i}}$, the object $\sigma$ is already pierced by $q_1$. This contradicts our algorithm.
Thus, the region $H\cap A_i$ contains at most one piercing point of $\A_{p,i}$.
To complete the proof, next, we will show that one can cover annular region $A_i$ using at most $\left(\Bigl\lceil2\left(1+\frac{1}{\alpha}\right)\Bigr\rceil^d-\Bigl\lfloor\frac{2}{\alpha}\Bigr\rfloor^d\right)$ interior disjoint hypercubes of side length $\frac{k}{(1+\alpha)^i}$.

{We can partition $\IR^d$ using hypercubes having side length $\frac{k}{(1+\alpha)^i}$. 
Formally, let $\Pi_d=\Big{\{}\alpha_1\frac{k}{(1+\alpha)^{i}} {\bf e}_1+\alpha_2\frac{k}{(1+\alpha)^{i}}{\bf e}_2+\ldots+\alpha_d\frac{k}{(1+\alpha)^{i}}{\bf e}_d\ \Big{|}\ (\alpha_1,\alpha_2,\ldots,\alpha_d)$ $\in \mathbb{Z}^d\Big{\}}$ be a lattice, where ${\bf e}_1,{\bf e}_2,\ldots,{\bf e}_d$ are the standard unit vectors (for an illustration of the lattice in $\IR$, $\IR^2$ and $\IR^3$, see Figure~\ref{fig:lattice}).}

\begin{clm}\label{clm:entire}
    For any point $q$ in $\IR^d$, there exists a point $r$ in $\Pi_d$ such that $d_{\infty}(q,r)\leq\frac{k}{2(1+\alpha)^i}$.
\end{clm}
\begin{proof}
{Notice that for any point $r\in\Pi_d$, each  coordinate of $r$ is an integral multiple of $\frac{k}{(1+\alpha)^i}$.}
For any point $q\in\IR^d$, for each $j\in[d]$, the $j$th coordinate of the point $q$ can  be uniquely written as $q(x_j)=z_j+y_j$, where $z_j\in\left(\frac{k}{(1+\alpha)^{i}}\right)\mathbb{Z}$ and $y_j\in\left[0,\frac{k}{(1+\alpha)^{i}}\right)$. Here,  by  $\beta\mathbb{Z}$ we mean the set $\left\{\beta z \ \Big{|}\ z\in\mathbb{Z}\right\}$.
Now, we define a point $r$ of $\Pi_d$ depending on the coordinates of $q$.
For each $j\in[d]$, we set the $j$th coordinate of $r$ as follows.
\[
r(x_j)=
\begin{cases}
    z_j,& \text{if $y_j\in\big[0,\frac{k}{2(1+\alpha)^{i}}\big)$}\\
   z_j+\frac{k}{(1+\alpha)^i}, & \text{if $y_j\in\big[\frac{k}{2(1+\alpha)^{i}},\frac{k}{(1+\alpha)^{i}} \big)$}.
\end{cases}
\]
As per the construction of {the point $r$, we have $|r(x_j)-q(x_j)|\leq \frac{k}{2(1+\alpha)^i}$ for each $j\in[d]$. As a result, $d_{\infty}(r,q)=\max_{j\in[d]}|r(x_i)-q(x_i)|\leq \frac{k}{2(1+\alpha)^i}$.}
\end{proof}

{Let $\Delta$ be the collection of hypercubes of side length $\frac{k}{(1+\alpha)^{i}}$ centered at the points $\Pi_d$.
Each hypercube of $\Delta$ is known as a \emph{cell}. 
As per the construction of the lattice, the distance between any two distinct points in $\Pi_d$ is at least $\frac{k}{(1+\alpha)^{i}}$. Thus, the cells in $\Delta$ are pairwise interior disjoint. Due to Claim~\ref{clm:entire}, the cells in $\Delta$ covers the entire $\IR^d$. }

{We create a copy $\Delta'$ of $\Delta$ by translating each cell of $\Delta$ by a translation vector such that one of the corners of a cell $c\in\Delta'$ coincides with one of the corners of $H_i$ {(see Figure~\ref{fig:tiling})}. For an instance, if the translation vector is $\Big(\left(p(x_1)+\frac{k}{\alpha (1+\alpha)^{i-1}}-\frac{k}{2(1+\alpha)^{i}}\right),\left(p(x_2)+\frac{k}{\alpha (1+\alpha)^{i-1}}-\frac{k}{2(1+\alpha)^{i}}\right),\ldots,$ $ \left(p(x_d)+\frac{k}{\alpha (1+\alpha)^{i-1}}-\frac{k}{2(1+\alpha)^{i}}\right)\Big)$, then the corner point $\Big(\left(p(x_1)+\frac{k}{\alpha (1+\alpha)^{i-1}}\right),\left(p(x_2)+\frac{k}{\alpha (1+\alpha)^{i-1}}\right),\ldots,$ $\left(p(x_d)+\frac{k}{\alpha (1+\alpha)^{i-1}}\right)\Big)$ of $H_i$ coincides with one of the corners of a cell $c\in\Delta'$.
 Let $\cal C$ be the minimal collection of cells of $\Delta'$ that cover the hypercube $H_i$.
It is easy to observe that the cardinality of $\cal C$ is at most~$\Bigl\lceil2\left(1+\frac{1}{\alpha}\right)\Bigr\rceil^d$.}
Notice that, as per the construction of the lattice, there exists a cell $c'\in\Delta'$, completely contained in $H_{i+1}$, whose one of the corners coincides with one of the corners of $H_{i+1}$. As a result, from the set $\cal C$, at least~$\Bigl\lfloor\frac{2}{\alpha}\Bigr\rfloor^d$ cells are contained in $H_{i+1}$.
Thus,  the union of at most~$\left(\Bigl\lceil2\left(1+\frac{1}{\alpha}\right)\Bigr\rceil^d-\Bigl\lfloor\frac{2}{\alpha}\Bigr\rfloor^d\right)$ hypercubes, each having side length $\frac{k}{(1+\alpha)^i}$, totally cover the annular region $A_i$. 
Hence, the lemma follows.
\end{proof}

Since $\cup \A_{p,i}=\A_p$, we have $|\A_p|\leq{\left(\Bigl\lceil2\left(1+\frac{1}{\alpha}\right)\Bigr\rceil^d-\Bigl\lfloor\frac{2}{\alpha}\Bigr\rfloor^d\right)}(m-1)+1$. As  the value of $m-1 \leq {\Bigl\lceil}\log_{1+\alpha}\left(\frac{2k}{\alpha}\right){\Bigr\rceil}$, the theorem follows.
\end{proof}

\noindent
The value of aspect$_{\infty}$ ratio is $\frac{1}{\sqrt{d}}$ for balls in $\IR^d$. As a result, we have the following.
\begin{corollary}\label{cor:ball}
For piercing  balls in $\IR^d$ with radius in the range $[1,k]$, $\AC$ achieves a competitive ratio of at most ~${\left(\Bigl\lceil2\left(1+\sqrt{d}\right)\Bigr\rceil^d-\Bigl\lfloor2\sqrt{d}\Bigr\rfloor^d\right)}{\lceil}\log_{1+\frac{1}{\sqrt{d}}}(2k\sqrt{d}){\rceil}+1$.
\end{corollary}

\noindent
\textbf{Unit Covering Problem in $\IR^d$}\\
In  Theorem~\ref{thm:alpha_inf}, if we fix the value of $k=1$, then due to Theorem~\ref{equivalence}, we have the following result for the unit covering problem.
\begin{corollary}\label{lem:unit}
For the unit covering problem using translates of a convex object $C$ in $\IR^d$, there exists a deterministic online algorithm whose competitive ratio is at most~${\left(\Bigl\lceil2\left(1+\frac{1}{\alpha}\right)\Bigr\rceil^d-\Bigl\lfloor\frac{2}{\alpha}\Bigr\rfloor^d\right)}{\Bigl\lceil}\log_{1+\alpha}(\frac{2}{\alpha}){\Bigr\rceil} +1$, where $\alpha$ is the aspect$_\infty$ ratio of {$C$}.
\end{corollary}

\subsection{For Similarly Sized Fat Objects in $\IR^3$}\label{ball-3-d}

\begin{theorem}
\label{thm:alpha-ball}
For piercing similarly sized $\alpha$-fat objects in $\IR^3$ having width in the range $[1,k]$, $\AC$ achieves a competitive ratio of at most~${\left(\left(1+\frac{1}{\sin(\theta/2)}\right)^3-1\right)}{\Bigl\lceil}\log_{1+x}(\frac{2k}{\alpha}){\Bigr\rceil}+1$, where $\theta=\frac{1}{2}\cos^{-1}\left(\frac{1}{2} +\frac{1}{1+\sqrt{1+4\alpha^{2}}}\right)$ and  $x=\frac{\sqrt{1+4{\alpha}^2}-1}{2}$.
\end{theorem}
\begin{figure}[htbp]
  \centering
    \begin{subfigure}[b]{0.32\textwidth}
          \centering
        \includegraphics[width=35mm]{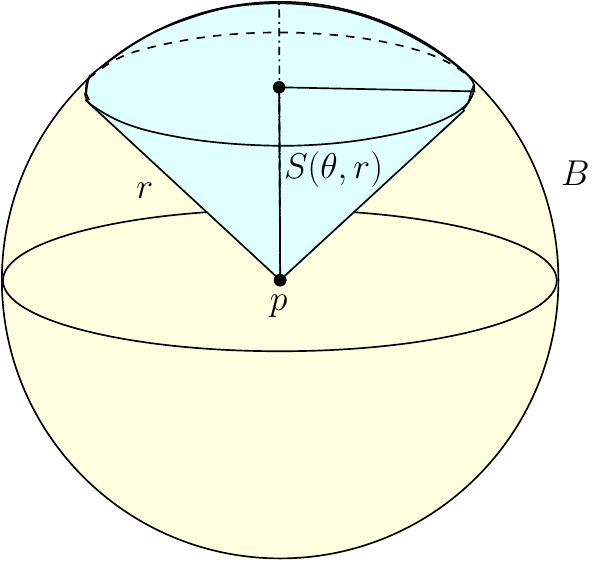}
    \caption{}
    \label{fig:ann_crea}
     \end{subfigure}
     \hfill
      \begin{subfigure}[b]{0.32\textwidth}
          \centering
    \includegraphics[width=35mm]{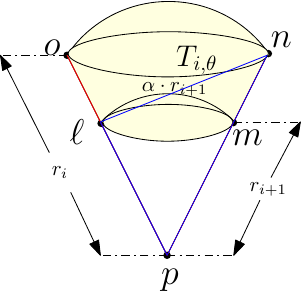}
    \caption{ }
    \label{fig:cone}
     \end{subfigure}
      \hfill
      \begin{subfigure}[b]{0.32\textwidth}
    \centering
    \includegraphics[width= 35mm]{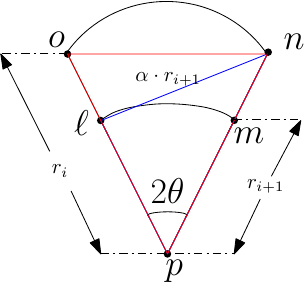}
    \caption{}
    \label{fig:cone_proj}
\end{subfigure}
  
 \caption{(a) Partitioning a ball $B$ of radius $r$ using spherical sector $S(\theta,r)$. (b) Description of a spherical sector $S(\theta,r_{i})$ and a spherical block $T_{i,\theta}$. (c) Projection of a spherical sector $S(\theta,r_{i})$.}
\end{figure}
\begin{proof}
 Let  $\I$ be the set of similarly sized $\alpha$-fat objects in $\IR^3$ presented to the algorithm.  Let $\A$ and $\OO$ be two piercing sets for $\I$ returned by $\AC$ and an offline optimal, respectively.  
Let $p$ be any piercing point of $\OO$. Let $\I_p\subseteq \I$ be the set of input  objects  pierced by the point $p$. Let $\A_p$ be the set of piercing points placed by our algorithm to pierce all the objects in $\I_p$.  To prove the theorem,  we will give an upper bound of $|\A_p|$.

Let us consider any point $a\in\A_p$. Since $a$ is the center of an $\alpha$-fat object $\sigma\in\I_p$ (containing the point $p$) having width at most $k$ (height at most $\frac{k}{\alpha}$), the distance between $a$ and $p$ is at most $\frac{k}{\alpha}$.
Therefore, a ball $B_1$ of radius $\frac{k}{\alpha}$, centered at $p$, contains all the points in $\A_p$.
Let $x=\frac{\sqrt{1+4{\alpha}^2}-1}{2}$ be a positive constant.
Let $B_i$ be a ball centered at $p$ having radius $r_i=\frac{k}{\alpha(1+x)^{i-1}}$, where $i\in[m]$ and $m$ is the smallest integer such that $\frac{k}{\alpha(1+x)^{m-1}}\leq\frac{1}{2}$. Note that $B_1,B_2,\ldots,B_m$ are concentric balls, centered at $p$. 
Let $\theta=\frac{1}{2}\cos^{-1}\left(\frac{1}{2} +\frac{1}{1+\sqrt{1+4\alpha^{2}}}\right)$ be a constant angle in
$(0,\frac{\pi}{10}]$. {Since $\theta=\frac{1}{2}\cos^{-1}\left(\frac{1}{2} +\frac{1}{1+\sqrt{1+4\alpha^{2}}}\right)$, and  $x=\frac{\sqrt{1+4{\alpha}^2}-1}{2}$, we have $\cos(2\theta)=\frac{(x+2)}{2(x+1)}$ and $x^2+x=\alpha^2$.}
Let $S(\theta,r_i)$ be a \emph{spherical sector}  obtained by taking the portion of the ball $B_i$ by  a conical boundary with the apex at the center $p$ of the ball and  $\theta$ as the half of the cone angle (for an illustration, see Figure~\ref{fig:ann_crea}). 
For any $i\in [m-1]$, let us define  the \emph{$i$th spherical block} $T_{i,\theta}=S(\theta,r_{i}) \setminus S(\theta,r_{i+1})$. 

\begin{clm}\label{claim:max_dist}
   The distance between any two points in $T_{i,\theta}$ is at most~{$\alpha r_{i+1}$}.
\end{clm}
\begin{figure}[htbp]
  \centering
   \begin{subfigure}[b]{0.24\textwidth}
          \centering
        \includegraphics[width=45mm]{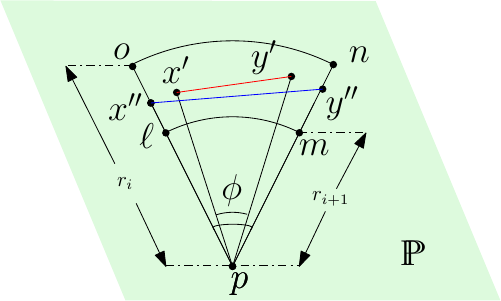}
    \caption{}
    \label{fig:plane}
     \end{subfigure}
      \hfill
      \begin{subfigure}[b]{0.24\textwidth}
    \centering
    \includegraphics[width= 30mm]{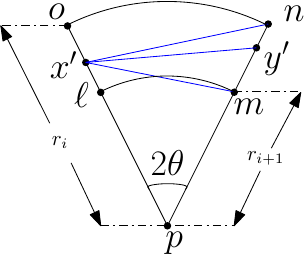}
    \caption{}
    \label{fig:des-1}
\end{subfigure}
     \hfill
    \begin{subfigure}[b]{0.24\textwidth}
          \centering
        \includegraphics[width=30mm]{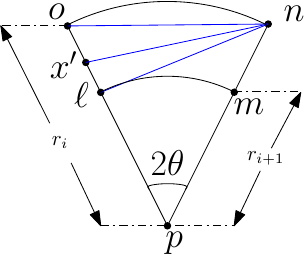}
    \caption{}
    \label{fig:des-2}
     \end{subfigure}
     \hfill
      \begin{subfigure}[b]{0.24\textwidth}
          \centering
    \includegraphics[width=30mm]{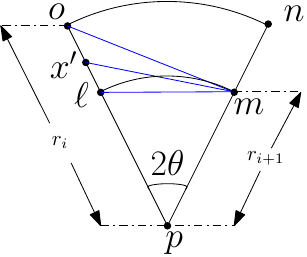}
    \caption{ }
    \label{fig:des-3}
     \end{subfigure}

 \caption{(a) Description of the plane $\mathbb{P}$. (b) Illustration of triangles $\triangle my'x'$ and $\triangle ny'x'$. (c) Illustration of triangles $\triangle oox'n$ and $\triangle \ell x'n$ (d) Illustration of triangles $\triangle ox'm$ and $\triangle x'\ell m$, in $T_{i,\theta}$.}
 \label{fig:projection}
\end{figure}

\begin{proof}
Let $x'$ and $y'$ be a farthest pair of points in $T_{i,\theta}$.
Let ${\mathbb P}$ be the plane passing through the points $p$, $x'$ and $y'$ (see Figure~\ref{fig:des-2}). Let $\ell,m,n,o$ be the corner points of $T_{i,\theta}\cap {\mathbb P}$ (see Figure~\ref{fig:projection}(b-d)). 
If $\angle x'py'<2\theta$, then we can find two points $x'', y''\in {\mathbb P}\cap T_{i,\theta}$ such that $d(p,x')=d(p,x'')$, $d(p,y')=d(p,y'')$ and $\angle x''py''=2\theta$ (see Figure~\ref{fig:plane}).
Due to the {law of cosines}, it is easy to observe that $d(x'',y'')> d(x',y')$. Therefore, the angle $\angle x'py'=2\theta$. In other words, $x'$ and $y'$ must lie on the line segment $\overline{\ell o}$ and $\overline{m n}$, respectively. In this case, we have $d(x',y')$ is at most $\max\{d(o,n),d(\ell,n)\}$.
To see this,  consider triangles $\triangle my'x'$ and $\triangle ny'x'$ (see Figure~\ref{fig:des-1}). Since the point $y'$ lies on the line segment $\overline{mn}$, either the angle $\angle x'y'n$ or the angle $\angle x'y'm$ is at least $\frac{\pi}{2}$. Now, due to the {law of sines}, we have $d(x',y')\leq \max\{d(x',n),d(x',m)\}$. For  similar reasons, we have $d(x',n)\leq \max\{d(o,n),d(\ell,n)\}$ (see Figure~\ref{fig:des-2}), and
$d(x',m)\leq \max\{d(o,m),d(\ell,m)\}=d(o,m)$ (see Figure~\ref{fig:des-3}). Due to the law of cosine, we have $d(\ell,n)=d(o,m)$. Therefore, $d(x',y')$ is at most $\max\{d(o,n),d(\ell,n),d(o,m)\}={\max}\{d(o,n),d(\ell,n)\}$.
 
{Now, we will prove that $d(o,n)=d(\ell,n)=\alpha r_{i+1}$.
  First, consider  the triangle $\triangle \ell p n$ (see Figure~\ref{fig:cone} and \ref{fig:cone_proj}). Here,  we have 
\begin{align*}
d^2(\ell,n)=&d^2(p,\ell)+d^2(p,n)-2d(p,\ell)d(p,n)\cos{(2\theta)} \tag{Due to the {law of cosines}}\\
 =&{r_{i+1}}^2+{r_i}^2-2{r_{i+1}}{r_i}\cos{(2\theta)} \tag{Since $d(p,\ell)=r_{i+1}$ and $d(p,n)=r_i$}\\
=&\left(\frac{k}{\alpha(1+x)^{i}}\right)^2+\left(\frac{k}{\alpha(1+x)^{i-1}}\right)^2-2\left(\frac{k}{\alpha(1+x))^{i}}\right)\left(\frac{k}{\alpha(1+x))^{i-1}}\right)\cos(2\theta) \tag{Since $r_{i+1}=\frac{k}{\alpha(1+x)^{i}}$ and $r_{i}=\frac{k}{\alpha(1+x)^{i-1}}$}\\
=&\left(\frac{k}{\alpha(1+x)^{i}}\right)^2 \left(1+(1+x)^2-2(1+x)\cos{(2\theta)}\right)\\
=&r_{i+1}^2\left(1+(1+x)^2-2(1+x)\cos{(2\theta)}\right) \tag{Since $r_{i+1}=\frac{k}{\alpha(1+x)^{i}}$}\\
=&r_{i+1}^2\left(1+(1+x)^2-(x+2)\right) \tag{Since $\cos(2\theta)=\frac{x+2}{2(x+1)}$}\\
=&r_{i+1}^2\left(1+1+x^2+2x-x-2\right)=r_{i+1}^2 \left(x^2+x\right)\\
=&\left(\alpha r_{i+1}\right)^2 \tag{Since $x^2+x=\alpha^2$}.
 \end{align*}}
 Now, consider the triangle $\triangle opn$ (see Figure~\ref{fig:cone} and \ref{fig:cone_proj}). Here, we have 
\begin{align*}
d^2(o,n)=&d^2(p,o)+d^2(p,n)-2d(p,o)d(p,n)\cos{(2\theta)} \tag{Due to the {law of cosines}}\\
=&{r_{i}}^2+{r_i}^2-2{r_{i}}{r_i}\cos{(2\theta)} \tag{Since $d(p,o)=d(p,n)=r_{i}$}\\
=&2{r_{i}}^2(1-\cos{(2\theta)})\\
=&2r_{i+1}^2(1+x)^2\left(1-\cos{(2\theta)}\right)\tag{Since  $r_i=r_{i+1}(1+x)$}\\
=&2r_{i+1}^2(1+x)^2\left(1-\frac{(x+2)}{2(x+1)}\right) \tag{Since  $\cos(2\theta)=\frac{x+2}{2(x+1)}$}\\
=&2r_{i+1}^2(1+x)^2\left(\frac{2(x+1)-(x+2)}{2(x+1)}\right)=r_{i+1}^2\left(x^2+x\right) \\
=&\left(\alpha r_{i+1}\right)^2 \tag{Since  $x^2+x=\alpha^2$}.
\end{align*}
Thus, we have $d(o,n)=d(\ell,n)=\alpha r_{i+1}$. Therefore {$\alpha r_{i+1}$ is the maximum distance between any two points in the region $T_{i,\theta}$.}
\end{proof}

\begin{clm}\label{claim:7}
For each $i\in [m]$, our algorithm places at most one piercing point in the spherical block $T_{i,\theta}$ to pierce any object in~$\I_p$.
 \end{clm}
\begin{proof}
Let $q_1$ be the first piercing point placed by $\AC$ in $T_{i,\theta}$. For a contradiction, let us assume that $\AC$ places another piercing point $q_2\in T_{i,\theta}$, where $q_2$ is the center of some object $\sigma\in \I_p$. Since $\sigma$ contains the point $p$, and the distance between $p$ and $q_2$ is at least {$r_{i+1}$}, the height and width of $\sigma$ are at least {$r_{i+1}$ and $\alpha  r_{i+1}$}, respectively.
Due to Claim~\ref{claim:max_dist}, the distance between any two points in the spherical block $T_{i,\theta}$ is at most ~$\alpha  r_{i+1}$. Therefore, the distance between $q_1$ and $q_2$ is at most $\alpha  r_{i+1}$. Since the width of $\sigma$ is at least $\alpha  r_{i+1}$, the object $\sigma$ is already pierced by $q_1$. This contradicts our algorithm.
Hence, $\AC$ places at most one piercing point in a spherical block $T_{i,\theta}$ to pierce objects in $\I_p$.
\end{proof}

Due to~\cite[Lemma~5.3]{DevroyeGL1997},
for any fixed $\theta\in(0,\pi/2)$, we need at most $\left(\left(1+\frac{1}{\sin(\theta/2)}\right)^3-1\right)$ spherical sectors $S{(1,\theta)}$ to completely cover the ball $B_1$. 
Combining this with Claim~\ref{claim:7},  we have $|\A_p|\leq{\left(\left(1+\frac{1}{\sin(\theta/2)}\right)^3-1\right)} m$. We can give a slightly better estimation {as follows}. Since the distance between any two points in the innermost ball $B_m$ is at most one, any object in $\I_p$  having center $q\in B_m$  contains the entire ball $B_m$. As a result, our online algorithm  places at most one piercing point in  $B_m$. Thus, $|\A_p|\leq{  \left(\left(1+\frac{1}{\sin(\theta/2)}\right)^3-1\right)}(m-1)+1\leq {\left(\left(1+\frac{1}{\sin(\theta/2)}\right)^3-1\right)}{\Bigl\lceil}\log_{1+x}(\frac{2k}{\alpha}){\Bigr\rceil}+1$. This completes the proof of the theorem.
 \end{proof}

Analogous to Theorem~\ref{thm:alpha-ball}, we have a similar result for $\alpha$-fat objects in $\IR^2$. Here, we consider circular sector ${C}(\theta,r)$ with central angle $2\theta$ instead of spherical sector $S(\theta,r)$ with central angle $2\theta$.

  \begin{theorem} \label{thm10}
  For piercing similarly sized $\alpha$-fat objects in $\IR^2$ having width  in the range $[1,k]$, $\AC$ achieves a competitive ratio of at most~{$\Bigl\lceil\frac{\pi }{\theta }\Bigr\rceil \Bigl\lceil\log_{1+x} (2k/\alpha)\Bigr\rceil  +1$},  where $\theta=\frac{1}{2}\cos^{-1}\left(\frac{1}{2} +\frac{1}{1+\sqrt{1+4\alpha^{2}}}\right)$ and  $x=\frac{\sqrt{1+4{\alpha}^2}-1}{2}$.
\end{theorem}

\begin{figure}[htbp]
          \centering
        \includegraphics[width=35mm]{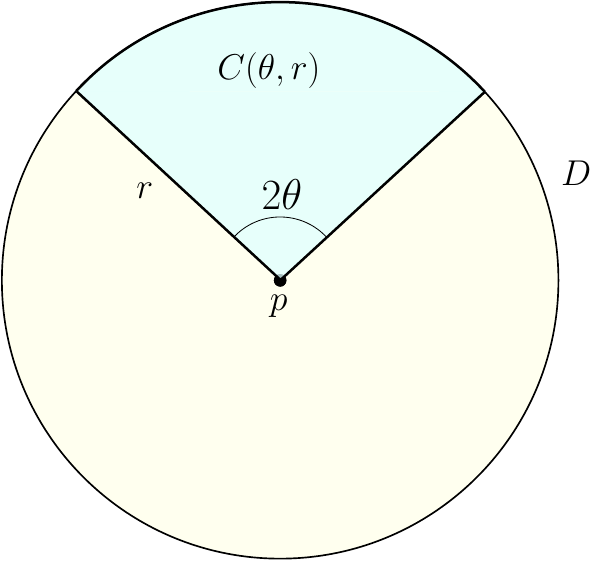}
    \caption{Partitioning the disk $D$ of radius $r$ using circular sector {$C(\theta,r)$}.}
    \label{fig:2d-ann_crea}
\end{figure}

\begin{proof}
 Let  $\I$ be the set of input $\alpha$-fat objects in $\IR^2$ presented to the algorithm.  Let $\A$ and $\OO$ be two piercing sets for $\I$ returned by $\AC$ and an offline optimal, respectively.  
Let $p$ be any piercing point of an offline optimal $\OO$ (see Figure~\ref{fig:2d-ann_crea}). Let $\I_p\subseteq \I$ be the set of input  objects  pierced by the point $p$. Let $\A_p$ be the set of piercing points placed by our algorithm to pierce all the objects in $\I_p$.  We will give an upper bound of $|\A_p|$ to prove the theorem.

Let us consider any point $a\in\A_p$. Since $a$ is the center of an $\alpha$-fat object $\sigma\in\I_p$ having width at most $k$ (height at most $\frac{k}{\alpha}$), the distance between $a$ and $p$ is at most $\frac{k}{\alpha}$.
Therefore, a disk $D_1$ of radius $\frac{k}{\alpha}$, centered at $p$, contains all the points in $\A_p$. Let $x=\frac{\sqrt{1+4{\alpha}^2}-1}{2}$ be a positive constant.
Let $D_i$ be a disk centered at $p$ having radius $r_i=\frac{k}{\alpha(1+x)^{i-1}}$, where $i\in[m]$ and $m$ is the smallest integer such that $\frac{k}{\alpha(1+x)^{m-1}}\leq\frac{1}{2}$. Note that $D_1,D_2,\ldots,D_m$ are concentric disks, centered at $p$. 
Let $2\theta=\cos^{-1}\left(\frac{1}{2} +\frac{1}{1+\sqrt{1+4\alpha^{2}}}\right)$ be a constant angle in
$(0,\frac{\pi}{5}]$.
Let ${C}(\theta,r_i)$ be a \emph{circular sector}  obtained by taking the portion of the disk $D_i$ by  a conical boundary with the apex at the center $p$ of the disk and $\theta$ as the half of the cone angle. 
For any $i\in [m-1]$, let us define  a \emph{$i$th circular block} {$T_{i,\theta}={C}(\theta,r_{i}) \setminus {C}(\theta,r_{i+1})$}.

Similar to the proof of {Claims~\ref{claim:max_dist} and~\ref{claim:7}}, one can prove the following.
\begin{clm}\label{2d-claim:max_dist}
   The distance between any two points in $T_{i,\theta}$ is at most~{$\alpha r_{i+1}$}.
\end{clm}

\begin{clm}\label{2d-claim:7}
For each $i\in [m]$, our algorithm places at most one piercing point in a circular block $T_{i,\theta}$ to pierce any object in~$\I_p$.
 \end{clm}

Since $\Bigl\lceil\frac{2\pi }{2\theta }\Bigr\rceil$  circular sectors will entirely cover $D_1$, the total number of piercing points $\A_p$ placed by our online algorithm to pierce objects in $\I_p$ is at most~$\Bigl\lceil{\frac{\pi}{\theta }}\Bigr\rceil\Bigl\lceil\log_{1+x}(2k/\alpha)\Bigr\rceil+1$.
Note that the disk of radius $\leq 1/2$ is common to all circular sectors. Only one piercing point is sufficient for that disk. Hence, one is not multiplied with {$\Bigl\lceil\frac{\pi }{\theta }\Bigr\rceil$}. Thus, the theorem follows.
 \end{proof}

Observe that for hypercubes in $\IR^d$, the value of $\alpha$ is $\frac{1}{\sqrt{d}}$. As a corollary to Theorem~\ref{thm:alpha-ball} (respectively, Theorem~\ref{thm10}), we have the following.
\begin{corollary}
For piercing similarly sized hypercubes in $\IR^3$ $($respectively, $\IR^2)$ with side length in the range $[1,k]$, $\AC$ achieves a
competitive ratio of at most~${909}{\Bigl\lceil}\log_{\frac{\sqrt{7}-\sqrt{3}}{2\sqrt{3}}}(2\sqrt{3}k){\Bigr\rceil}+1$ $\Big($respectively, ${12\Bigl\lceil}\log_{\frac{\sqrt{7}-\sqrt{3}}{2\sqrt{3}}}(2\sqrt{2}k){\Bigr\rceil}+1\Big)$.
\end{corollary}

Observe that the value of $\alpha$ is $1$ for balls. As a corollary to Theorem~\ref{thm10}, we have the following upper bound that gives a better bound than Corollary~\ref{cor:ball} for balls in $\IR^2$. 
\begin{corollary}
For piercing  balls in $\IR^2$  with radius in the range $[1,k]$, $\AC$ achieves a 
competitive ratio of at most~${10\Bigl\lceil}\log_{\frac{\sqrt{5}-1}{2}}(2k){\Bigr\rceil}+1$. 
\end{corollary}

\subsection{For Homothetic Hypercubes in $\IR^d$}
Since the value of aspect$_{\infty}$ ratio is $1$ for (axis-aligned) hypercube and the fact that a hypercube with a side length in the range $[1,k]$ is equivalent to a hypercube with a width in the range $[\frac{1}{2},\frac{k}{2}]$, as a  corollary to Theorem~\ref{thm:alpha_inf}, $\AC$ has
a competitive ratio of at most~$2^d\left(2^{d} -1\right){\lceil}\log_{2} k{\rceil}+1$  for piercing homothetic hypercubes having a side length in the range $[1,k]$. Here, we propose an algorithm that has a better competitive ratio.

Let $\sigma\subset \IR^d$ be a hypercube having side length $\ell$. We partition the hypercube $\sigma$ into $2^d$ smaller sub-hypercubes, each having side length $\ell/2$. Let $P_\sigma^j$ be a $j$th sub-hypercube of $\sigma$, where $j\in[2^d]$.
\begin{observation}\label{obs:3}
For each $j\in[2^d]$, if we place points at the vertices of $P_\sigma^j$, then $\sigma$ contains exactly $3^d$ distinct points.
\end{observation}
\begin{proof}
   Without loss of generality, let the center of $\sigma$ coincide with the origin, and the side length of $\sigma$ is two. Note that the sub-hypercube $P_\sigma^j$ is a hypercube of side length one. Thus, the corner points of $P_\sigma^j$ are integer points. To prove the observation, it suffices to calculate the number of distinct integer points contained in the hypercube $\sigma$. Since $\sigma$ is centered at the origin and has side length two, each coordinate of any integer point in $\sigma$ has three possible values from $\{-1,0,1\}$. As a result, $\sigma$ contains exactly $3^d$ integer points. 
\end{proof}

\noindent
\textbf{\AV:} 
Let $\A$ be the piercing set maintained by our algorithm to pierce the incoming hypercubes. Initially, $\A=\emptyset$. Our algorithm does the following on receiving a new input hypercube $\sigma$.
\begin{itemize}
  \item If the existing piercing set pierces $\sigma$, do nothing.
  \item Otherwise,
  \begin{itemize}
      \item if $\sigma$ is a unit hypercube, our online algorithm adds all the vertices of $\sigma$ to $\A$.
      \item else, our online algorithm adds all vertices of $P_\sigma^j$ to $\A$, for all $j\in[2^d]$.
  \end{itemize}
  \end{itemize}
Due to Observation~\ref{obs:3}, upon the arrival of a hypercube in $\IR^d$, $\AV$ adds at most $3^d$ distinct points.

\begin{theorem}\label{thm:d-hyp}
For piercing homothetic  hypercubes in $\IR^d$ having side length in the range $[1,k]$, \textsc{Algorithm-Vertex} has a competitive ratio of at most~$3^d{\lceil}\log_2 k{\rceil}+2^d$.
\end{theorem}
\begin{proof}
{Let  $\I$ be a set of input hypercubes in $\IR^d$ presented to the algorithm. Let $\A$ and $\OO$ be two piercing sets for $\I$ returned by our algorithm and an offline optimal, respectively. Let $p$ be any piercing point of an offline optimal $\OO$. Let $\I_p\subseteq \I$ be the set of input objects pierced by the point $p$. Let $\A_p$ be the set of piercing points placed by our algorithm to pierce all the objects in $\I_p$.  We will give an upper bound of $|\A_p|$ to prove the theorem.}
Let $H_i\subset\I_p$ be the collection of hypercubes in $\I_p$ having side length $\left(\frac{k}{2^i},\frac{k}{2^{i-1}}\right]$, where $i\in [{\lceil}\log_2 k{\rceil}]$, and let $H_{{\lceil}\log_2 k{\rceil}+1}$ be the collection of all hypercubes in $\I_p$ having side length one. Note that $\I_p=\cup_{j=1}^{{\lceil}\log_2 k{\rceil}+1} H_j$.

 For each $j\in[{\lceil}\log_2{k}{\rceil}]$, we claim that our online algorithm places at most $3^d$ piercing points for all hypercubes in $H_j$. 
 Without loss of generality, let us assume that $\sigma\in H_j$ is the first hypercube that is not pierced upon its arrival. To pierce $\sigma$, our online algorithm adds vertices of all $P_\sigma^j$, where $j\in[2^d]$ to $\A_p$. Let $P_\sigma^t$ contains the point $p$, where $t\in[2^d]$. Let $\sigma'(\neq \sigma) \in  H_j$ be any hypercube.  Observe that the side length of $P_\sigma^t$ is at most $\frac{k}{2^j}$ and the side length of  $\sigma'\in H_j$ is at least $\frac{k}{2^j}$, and $\sigma'\cap P_\sigma^t \neq \emptyset$. As a result, $\sigma'$ contains at least one vertex of $P_\sigma^t$. Hence, our algorithm does not add any point to $\A_p$ for $\sigma'$. In other words, due to Observation~\ref{obs:3}, our online algorithm adds $3^d$ distinct points to $\A_p$ for all hypercubes in $H_j$, where $j\in[{\lceil}\log_2{k}{\rceil}]$.

 Now, we will show that our online algorithm places at most $2^d$ piercing points for all hypercubes in $H_{{\lceil}\log_2 k{\rceil}+1}$.
Let us assume that the unit hypercube $\sigma\in H_{{\lceil}\log_2 k{\rceil}+1}$ is the first hypercube that is not pierced upon its arrival. To pierce $\sigma$, our online algorithm adds all the vertices of $\sigma$ to $\A_p$. Note that any hypercube $\sigma'$ in $H_{{\lceil}\log_2 k{\rceil}+1}$ is a translated copy of $\sigma$, and $p$ is the common intersection point between them. Thus, the hypercube $\sigma'$ contains at least one corner point of $\sigma$. Consequently, our algorithm adds at most $2^d$ piercing points to $\A_p$  to pierce all hypercubes in $H_{{\lceil}\log_2 k{\rceil}+1}$.

As a result, $\A_p$ contains at most $3^d{\lceil}\log_2k{\rceil}+2^d$ points. Therefore, the theorem follows.
\end{proof}

\noindent
\textbf{Remark.}  Note that for piercing {translates of a hypercube}, $\AV$ matches the best-known result for the unit covering problem using translates of a hypercube~\cite{ChanZ09}.

\section{Lower Bound}\label{sec:lb}
To obtain a lower bound, we think of a game between two players: Alice and Bob. Here, Alice plays the role of an adversary, and Bob plays the role of an online algorithm. In each round of the game, Alice presents an object such that Bob needs to place a new piercing point, i.e., the object does not contain any of the previously placed piercing points. To obtain a lower bound of the competitive ratio of $\Omega(z)$, it is enough to show that Alice can present a sequence of $z$ nested objects in a sequence of $z$ consecutive rounds of the game such that an offline optimum algorithm uses only one point to pierce all of these objects. Throughout the section, $0<\epsilon<\frac{1}{4}$ is an arbitrary constant close to 0. For example, $\epsilon=10^{-10}$ will also work.
We consider similarly sized $\alpha$-fat objects in $\IR^2$, followed by homothetic hypercubes in $\IR^d$. 
\subsection{Similarly Sized Fat Objects in $\IR^2$}
\begin{theorem}\label{thm:lowerBoundPolygon}
{For $\alpha\in(0,1]$ and $k\geq 1$}, the  competitive ratio of every deterministic online algorithm  for piercing similar sized $\alpha$-fat objects in $\IR^2$ having width in the range $[1,k]$ is at least ${\Bigl\lfloor\log_{\frac{2+\epsilon}{\alpha}}}{k}{\Bigr\rfloor}+1$, {where $0<\epsilon<\frac{1}{4}$ is a sufficiently small constant close to $0$}.
\end{theorem}
\begin{proof}
To prove the lower bound, Alice adaptively construct a sequence of $m={\Bigl\lfloor}\log_{\frac{2{+\epsilon}}{\alpha}}{k}{\Bigr\rfloor}+1$ objects, each with aspect ratio  $\alpha$ and width in the range $[1,k]$ such that Bob needs at least $m$ piercing points to pierce them, while an offline optimum needs just one point. 
Let $w(\sigma)$ and $h(\sigma)$ be the width and height of an object $\sigma$. Let $\sigma_1$, having width $w(\sigma_1)=k$, be the first input object  presented by Alice. For the sake of simplicity, let us assume that the center $c_1$  of $\sigma_1$ coincides with the origin. All remaining objects $\sigma_2,\sigma_3,\ldots,\sigma_{m}$ are presented by Alice adaptively depending on the position of the piercing points $p_1,p_2,\ldots,p_{m-1}$ placed by Bob.  For $i=1,2,\ldots,m={\Bigl\lfloor}\log_{\frac{2{+\epsilon}}{\alpha}}{k}{\Bigr\rfloor}+1$, we maintain the following two invariants.
\begin{itemize}
    \item [(1)] The object $\sigma_i$ having width  $k\left(\frac{\alpha}{2+\epsilon}\right)^{i-1}$ is not pierced by any of the previously placed piercing point $p_j$, where {$j\in[i-1]$};
    \item [(2)] the object $\sigma_i$ is totally contained in the object $\sigma_{i-1}$.
\end{itemize}

 Invariant (1) ensures that Bob needs $m$ piercing points, while invariant (2) ensures that all the objects $\sigma_1,\sigma_2,\ldots,\sigma_{m}$ can be pierced by a single point.
 For $i=1$, both invariants hold trivially. At the end of round $i=1,2,\ldots,m-1$, assume that both invariants hold. {At the end of round $i=m$}, we will show that both invariants hold. Depending on the position of the previously placed piercing point $p_{i-1}$, in the $i$th round of the game, an object $\sigma_{i}$,   having width $w(\sigma_{i})=k\left(\frac{\alpha}{2+\epsilon}\right)^{i-1}$ is presented by Alice to Bob. 
 The center  $c_{i}$ of $\sigma_{i}$  is defined as the following.
\begin{equation*}\label{eqn}
c_{i}=
\begin{cases}
    c_{i-1}+\frac{w(\sigma_{i-1})}{2},& \text{if } p_{i-1}(x_1)\leq c_{i-1}(x_1),\\
      c_{i-1}-\frac{w(\sigma_{i-1})}{2},& \text{otherwise } (i.e., p_{i-1}(x_1)>c_{i-1}(x_1)),
    \end{cases}      
\end{equation*}
where $p_{i-1}(x_1)$ and $c_{i-1}(x_1)$ denotes the first coordinate of $p_{i-1}$ and $c_{i-1}$, respectively.

First, we show that $\sigma_{i}$ is totally contained in $\sigma_{i-1}$.
 Observe that, depending on the position of $p_{i-1}$, the center of $\sigma_{i}$ is either $c_{i-1}+\frac{w(\sigma_{i-1})}{2}$ or $c_{i-1}-\frac{w(\sigma_{i-1})}{2}$. In both cases, we have
$d(c_{i-1},c_{i})=\frac{w(\sigma_{i-1})}{2}$. On the other hand,  $h(\sigma_{i}) = \frac{k}{2+\epsilon}\left(\frac{\alpha}{2+\epsilon}\right)^{i-2}< \frac{k}{2}\left(\frac{\alpha}{2+\epsilon}\right)^{i-2}=\frac{w(\sigma_{i-1})}{2}$.
Hence, $\sigma_{i}$ is totally contained in $\sigma_{i-1}$. Thus, invariant~(2) is maintained.

Note that  $d(p_{i-1},c_{i})$ is greater than the height of $\sigma_{i}$ since
\begin{align*}
    d(p_{i-1},c_{i})
    >\frac{w(\sigma_{i-1})}{2}=\frac{k}{2}\left(\frac{\alpha}{2+\epsilon}\right)^{i-2}
    >&\frac{k}{2+\epsilon}\left(\frac{\alpha}{2+\epsilon}\right)^{i-2}=h(\sigma_{i}).
\end{align*}
The first inequality follows from the definition of $c_{i}$.
Thus, $\sigma_{i}$ does not contain the point $p_{i-1}$.
  Due to the induction hypothesis, $\sigma_{i-1}$ does not contain any of the previously placed piercing point $p_j$ for $j\in[{i-2}]$, and from invariant~(2), we know that $\sigma_{i}$ is contained in $\sigma_{i-1}$. Hence, $\sigma_{i}$ does not contain any of the previously placed piercing point $p_j$, for $j\in[i-1]$. Thus, invariant~(1) is maintained.
  
 Note that when $m > \log_{\frac{2+\epsilon}{\alpha}}{k} +1$, the  width of the object $\sigma_{m}$, i.e., $k\left(\frac{\alpha}{2+\epsilon}\right)^{m-1}$ is  less  than    one. 
 In other words,  for any $\log_{\frac{2+\epsilon}{\alpha}}{k}<m\leq \log_{\frac{2+\epsilon}{\alpha}}{k} +1$, we can construct the input sequence satisfying both invariants.
Since the value of $m$ needs to be integer, we have $m=\Bigl\lfloor\log_\frac{2+\epsilon}{\alpha}{k}\Bigr\rfloor+1$. Hence, the theorem follows.
\end{proof}

We know that the value of $\alpha$ is $1$ for balls in $\IR^2$. As a result, we have 
\begin{corollary}
For $k\geq 1$, the competitive ratio of every deterministic online algorithm for piercing balls in $\IR^2$ having  radius in the range $[1,k]$ is at least $\lfloor \log_{2+\epsilon} k \rfloor-1$, where $0<\epsilon<\frac{1}{4}$ is a sufficiently small constant close to $0$.
\end{corollary}

 \subsection{Homothetic Hypercubes in $\IR^d$}\label{I(C)}
In this subsection, all the hypercubes are axis-aligned.
First, we establish the following essential ingredients to prove our main result.

\begin{lemma}\label{claim:lb_hyper}{$\GS(r)$ $(GSS(r))$:}

For any $r\geq 1$, in $d$ consecutive rounds of the game, Alice can adaptively present $d$ hypercubes $\sigma_{1},\sigma_{2},\ldots,\sigma_{d}$, each having side length $r$ {contained inside a hypercube $S$ having side length $r(2+\epsilon)$}, such that
\begin{itemize}
\item[(i)]  Bob needs to place $d$   points  to pierce them;
\item[(ii)]  the common intersection region $Q=\cap_{i=1}^{d} \sigma_i$ is nonempty; 
\item[(iii)]  moreover, $Q$  contains an \emph{empty} hypercube $E$, of side length at least $\frac{r}{2+\epsilon}$, not containing any of the $d$  piercing points placed by Bob, {where $0<\epsilon<\frac{1}{4}$ is a sufficiently small constant close to $0$}.
\end{itemize}
  
\end{lemma}
We refer to $d$ consecutive rounds of the game satisfying the above lemma as a $GSS(r)$. 
{The hypercube $S$ is denoted as the \emph{starter}}, and the hypercube $E$ is the \emph{empty hypercube} since it does not contain any of the piercing points placed by Bob in $GSS(r)$.
\begin{proof} 
{Throughout the proof of this lemma, for the sake of simplicity, we assume that the center of $S$ is the origin.}
Let $\sigma_{1}$ be a hypercube of side length $r$ presented by Alice in the first round of the game. {Let  the center of $\sigma_1$ be the same as  the center of $S$.}
 Alice presents the remaining  hypercubes $\sigma_2,\ldots, \sigma_d$, adaptively 
 depending on  Bob's moves. 
  We maintain the following two invariants:  For $i=1,\ldots, d$, when Alice presents  hypercubes  $\sigma_1,\ldots, \sigma_i$ each of side length $r$, and Bob presents piercing points $p_1, \ldots, p_i$, 

\begin{itemize}
    \item[(I)] the hypercube $\sigma_{i}$ is contained inside $S$ and $\sigma_i$ is not pierced by any of the previously placed piercing point $p_{j}$, where {$j\in[i-1]$};
    \item[(II)]  the common intersection region $Q_i=\cap_{j=1}^{i} \sigma_j$ is a  hyperrectangle whose first $(i-1)$ sides are of length $\frac{r}{2+\epsilon}$ each, and  each of the remaining sides are of length  $r$. Moreover, $Q_i$ does not contain any points $p_j$, where {$j\in[i-1]$}.
    \end{itemize}

 An illustration of the planar version of the
game appears in Figure~\ref{fig:sq_lb}.
For $i=1$, both the invariants trivially hold.
For $i=1,2,\ldots,d-1$, assume that both invariants hold. Now, for $i=d$, we will show that both invariants hold. Let us define a translation vector ${\bf v}_i\in \IR^d$ as
${\bf v}_i=\Big(s(1)\left(\frac{r}{2-\epsilon}\right),s(2)\left(\frac{r}{2-\epsilon}\right),\ldots,s(i-1)\left(\frac{r}{2-\epsilon}\right)$ $,0,\ldots,0\Big)$, where $s(j)$, for $j\in[i-1]$, is defined as follows

  \[
  s(j)= 
\begin{cases}
    +1,& \text{if}\ p_{j}(x_j)<0,\text{ where}\ p_{j}(x_j)\text{ is $jth$ coordinate of $p_{j}$}, \\
    -1,              & \text{otherwise.}
\end{cases}
\]

We define $\sigma_i=\sigma_1 +{\bf v}_i$. First, we will prove that $\sigma_i$ is contained inside the hypercube $S$. Since $\sigma_i=\sigma_1 +{\bf v}_i$, the $L_{\infty}$ distance between the center of $\sigma_1$ and any point in $\sigma_i$ is at most $\frac{r+\frac{r}{2-\epsilon}}{2}$, which is strictly less than $\frac{r(2+\epsilon)}{2}$. As a result, $\sigma_i$ is contained inside $S$. Now, we will show that $\sigma_i$ does not contain the point $p_j$.  For any  $j<i$, due to the definition of the $j$th component of the translation vector ${\bf v}_i$, the hypercube $\sigma_i$ does not contain the point $p_j$. Hence, invariant (I) is maintained.

Note that $\sigma_1\cap \sigma_i$ is a  hyperrectangle whose first $(i-1)$ sides are  of length $\frac{r}{2+\epsilon}$ each, and each of the remaining sides is of length  $r$. On the other hand, from the assumption, we know that $Q_{i-1}$ is a  hyperrectangle whose  first $(i-2)$ sides are of length $\frac{r}{2+\epsilon}$ each, and  each of the remaining sides are of length  $r$. Therefore, $Q_i=Q_{i-1}\cap \sigma_i$ is a hyperrectangle whose  first $(i-1)$ sides are of length $\frac{r}{2+\epsilon}$ each, and  each of the remaining sides are of length  $r$. Since $\sigma_i$ does not contain any of the point $p_j$ where $j<i$, the  hyperrectangle $Q_i\subseteq \sigma_i$ also does not contain any of them.  Therefore, invariant (II) is maintained.

At the end of the $d$th round of the game, we have the  hyperrectangle $Q=Q_d$ (whose $d$th side is only of length $r$, other sides are of length $\frac{r}{2+\epsilon}$ each) that does not contain any of the point $p_j$ where $j<d$, but it may contain the point $p_d$. Depending on the $d$th coordinate of the point $p_d$, we can construct a  hypercube $E \subset Q$ of side length $\frac{r}{2+\epsilon}$ that does not contain the point $p_d$. Hence, the lemma follows.
\end{proof}
\begin{figure}
    \centering
    \includegraphics[width=60 mm]{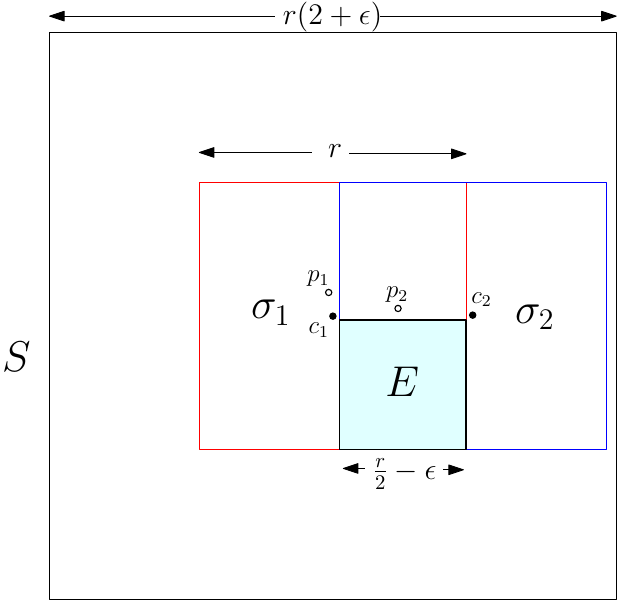}
    \caption{Illustration of  the $GSS(r)$ for $d=2$. 
Initially, Alice presented hypercube $\sigma_{1}$ of side length $r$ in the first round of the game.  Then, Alice presents the second hypercube $\sigma_2$ adaptively,
 depending on  Bob's moves, such that $p_1 \notin \sigma_2$. Also, the common intersection region $\cap_{j=1}^{2} \sigma_j$ is a rectangle whose sides are  of length $\frac{r}{2+\epsilon}$ and $r$.  Moreover, the intersection region contains an \emph{empty} hypercube $E$, of side length at least $\frac{r}{2+\epsilon}$, not containing $p_1$ and $p_2$.}
  \label{fig:sq_lb}
\end{figure}
Now, we prove the following main theorem.

\begin{theorem}\label{thm_sq_d}
For $k\geq 1$, the competitive ratio of every deterministic online algorithm for  piercing homothetic hypercubes in $\IR^d$ with side length in the range $[1,k]$  is at least~$d{\Bigl\lfloor}\log_{(2+\epsilon){^2}}k{\Bigr\rfloor}+2^d$, where $0<\epsilon<\frac{1}{4}$ is a sufficiently small constant close to $0$.
\end{theorem}
\begin{proof}
We maintain the following two invariants:  For $i=1,\ldots, z={d{\Bigl\lfloor}\log_{(2+\epsilon){^2}}k{\Bigr\rfloor}+2^d}$, when Alice presents  hypercubes  $\sigma_1,\ldots, \sigma_i$ and Bob presents piercing points 
 $p_1, \ldots, p_i$, 
 \begin{itemize}
    \item[(III)] the  hypercube $\sigma_i$ is not pierced by any of the previously placed piercing point $p_j$, where {$j\in[i-1]$};
    \item[(IV)] the intersection region $\cap_{j=1}^{i} \sigma_j$ is nonempty.
    \end{itemize}
 Invariant~(III) implies that Bob is forced to place $z$ piercing points, while invariant~(IV) ensures that all the hypercubes $\sigma_1,\sigma_2,\ldots,\sigma_{z}$ can be pierced by a single piercing point.

We evoke the $GSS(r)$ for $t={\Bigl\lfloor\log_{(2+\epsilon)^2} k\Bigr\rfloor}$ time with different values of {$r=k, \frac{k}{\left(2+\epsilon\right)^2}, \ldots,\frac{k}{\left(\left(2+\epsilon\right)^2\right)^{t-1}}$}.
More specifically, for any $t'\in[t-1]$,
the empty hypercube $E_{t'}$ having side length $\frac{k}{(2+\epsilon)^{2t'-1}}$ generated by the $t'$th GSS game, $GSS\left({\frac{k}{(2+\epsilon)^{2(t'-1)}}}\right)$, plays the role of the starter for the $(t'+1)$th evocation of the GSS game, $GSS\left({\frac{k}{(2+\epsilon)^{2t'}}}\right)$.
Since all the $d$ hypercubes of $(t'+1)$th game lie in the interior of the empty hypercube $E_{t'}$, it is straightforward to see that the set of $td$  hypercubes $\sigma_1,\sigma_2,\ldots,\sigma_{td}$ and set of $td$ points $p_1,p_2,\ldots, p_{td}$ satisfy both invariants (III) and (IV).
At the end of the $t$th $GSS$, the empty  hypercube $E_{t}$ has a side length of at least {$(2+\epsilon)$}.
Observe that, Due to Corollary~\ref{corr_illum}, Alice can adaptively present $2^d$  hypercubes, each of side length one, such that  Bob needs to place $2^d$ different piercing points, whereas an offline optimum needs one point to pierce all these $2^d$ hypercubes. Since all these $2^d$ hypercubes having side length one has a common intersection region, all of these hypercubes can be placed inside a hypercube of side length two.  As a result, Alice can adaptively present $2^d$ hypercubes each of side length one inside $E_t$ such that both invariants (III) and (IV) are maintained.
This completes the proof of the theorem.
\end{proof}

\section{Conclusion}\label{Conclusion}
No deterministic online algorithm can obtain a competitive ratio lower than $\Omega(n)$ for piercing $n$ intervals.  Due to this pessimistic result, we restricted our attention to similarly sized objects.  
We propose upper bounds for piercing similarly sized $\alpha$-aspect$_{\infty}$ fat objects in higher dimensions. By placing  specific values of $\alpha$, the tight asymptotic bounds can be obtained for similarly sized objects (like intervals, squares, disks, and $\alpha$-aspect polygons) in $\IR$ and $\IR^2$.  There are  gaps between the lower and the upper bounds in higher dimensions. We propose bridging these gaps as a future direction of research.
We only consider the deterministic model, which raises the question of whether randomization helps.

\subsection*{Acknowledgement} 
{The authors would like to thank anonymous reviewers for their constructive feedback that enhanced the quality of the article.}

\section*{Conflict of Interest}
The authors declare that no financial and non-financial competing interests are relevant to this article's content.

\bibliographystyle{plain}
\bibliography{references}
\end{document}